\documentclass[11pt]{article}
\usepackage{a4wide}
\usepackage{amsmath, amsfonts, amssymb, amsthm, complexity, hyperref, authblk}

\newcommand{\N}{\mathbb N}

\renewcommand{\A}{\mathbb A}

\newcommand{\F}{\mathbb F}

\newcommand{\abs}[1]{\lvert #1 \rvert}
\newcommand{\abrace}[1]{\left< #1 \right> }
\newcommand{\cbrace}[1]{\left\{ #1 \right\} }
\newcommand{\sqbrace}[1]{\left[ #1 \right] }
\newcommand{\ceil}[1]{\lceil #1 \rceil}
\newcommand{\comment}[1]{} 

\renewcommand{\L}{{\boldsymbol L}}
\newcommand{\f}{{\boldsymbol f}}	
\newcommand{\w}{\mathrm{w}}

\renewcommand{\M}{M}
\newcommand{\Ml}{M_\ell}
\newcommand{\T}{T}
\newcommand{\Tl}{T_\ell}
\newcommand{\Tls}{T_{\ell,k}}

\newcommand{\e}{{\boldsymbol e}} 
\newcommand{\x}{{\boldsymbol x}} 
\newcommand{\y}{{\boldsymbol y}} 
\newcommand{\z}{{\boldsymbol z}} 
\renewcommand{\a}{{\boldsymbol a}}  
\renewcommand{\b}{{\boldsymbol b}}  
\newcommand{\s}{{\boldsymbol s}}  
\newcommand{\h}{{\boldsymbol h}}  
\newcommand{\bfgamma}{{\boldsymbol \gamma}}
\renewcommand{\v}{{\boldsymbol v}}  

\newcommand{\ROABP}{{\rm ROABP}}

\newcommand{\xa}{\x^{\a}}  
\newcommand{\ya}{\y^{\a}}  
\newcommand{\xb}{\x^{\b}}  
\newcommand{\zkb}{{\z_k^{\b}}}  
\newcommand{\yka}{{\y_k^{\a}}}  

\newcommand{\mydot}{\boldsymbol{\cdot}} 
\newcommand{\row}[2]{#1(#2,\mydot)}           
\newcommand{\column}[2]{#1(\mydot,#2)}           
\newcommand{\entry}[3]{#1(#2,#3)}           
\newcommand{\lis}[3]{{#1}_1 #2 {#1}_2 #2\dots #2 {#1}_{#3}}
\newcommand{\wt}[1]{(2w)^{2^{#1}}}

\DeclareMathOperator{\supp}{supp}
\DeclareMathOperator{\max1}{max}

\DeclareMathOperator{\Span}{span}
\DeclareMathOperator{\coeff}{coeff}
\DeclareMathOperator{\rank}{rank}
\DeclareMathOperator{\lc}{lc}

\DeclareMathOperator{\spanning}{span}
\DeclareMathOperator{\depending}{depend}

\newcommand{\coeffset}[3]{{#1}_{(#2,#3)}} 

\newtheorem{theorem}{Theorem}[section]
\newtheorem{lemma}[theorem]{Lemma}

\newtheorem{corollary}[theorem]{Corollary}
\newtheorem{definition}[theorem]{Definition}
\newtheorem{claim}[theorem]{Claim}

\title{Deterministic Identity Testing for Sum of Read-Once Oblivious Arithmetic Branching Programs}
\author[1]{Rohit Gurjar\thanks{\texttt{rgurjar@cse.iitk.ac.in}, 
supported by TCS PhD research fellowship}}
\author[1]{Arpita Korwar\thanks{\texttt{arpk@cse.iitk.ac.in}}}
\author[1]{Nitin Saxena\thanks{\texttt{nitin@cse.iitk.ac.in}, supported by DST-SERB}}
\author[2]{Thomas Thierauf\thanks{\texttt{thomas.thierauf@htw-aalen.de}, supported by DFG grant TH 472/4-1}}
\affil[1]{Department of Computer Science and Engineering, IIT Kanpur, India
}
\affil[2]{Aalen University, Germany
}

\date{\vspace{-1.5cm}}

\begin{document}

\maketitle

\begin{abstract}
A {\em read-once oblivious arithmetic branching program (ROABP)\/} is an arithmetic branching program (ABP)
where each variable occurs in at most one layer. 
We give the first polynomial time whitebox identity test for a polynomial computed by a sum of constantly many
ROABPs.
We also give a corresponding blackbox algorithm with quasi-polynomial time complexity~$n^{O(\log n)}$.
In both the cases, our time complexity is double exponential in the number of ROABPs.

ROABPs are a generalization of set-multilinear depth-$3$ circuits.
The prior results for the sum of constantly many set-multilinear depth-$3$ circuits
were only slightly better than brute-force, i.e.\ exponential-time.

Our techniques are a new interplay of three concepts for ROABP:
low evaluation dimension, basis isolating weight assignment and low-support rank concentration.
We relate basis isolation to rank concentration and extend it to a sum of two ROABPs using evaluation dimension
(or partial derivatives).
\end{abstract}


\section{Introduction}
Polynomial Identity Testing (PIT) is the problem of testing whether
a given $n$-variate polynomial is identically zero or not.
The input to the PIT problem
may be in the form of arithmetic circuits or
arithmetic branching programs (ABP).
They are the arithmetic analogues of boolean circuits and
boolean branching programs, respectively.
It is well known that PIT can be solved in randomized polynomial time, see e.g.~\cite{Sch80}.
The randomized algorithm just evaluates the polynomial at random points;
thus, it is a \emph{blackbox\/} algorithm.
In contrast,
an algorithm is a \emph{whitebox\/} algorithm if it 
looks inside the given circuit or branching program.
We consider both, whitebox and blackbox algorithms.

Since all problems with randomized polynomial-time solutions
are conjectured to have deterministic polynomial-time algorithms,
we expect that such an algorithm exists for PIT.
It is also known that any sub-exponential time algorithm for PIT
implies a lower bound~\cite{KI03, Agr05}.
See also the surveys~\cite{Sax09, Sax14, SY10}.

An efficient deterministic solution for PIT is known only for very restricted
input models, for example, sparse polynomials~\cite{BOT88,KS01}, 
constant fan-in depth-$3$ ($\Sigma \Pi \Sigma$) circuits \cite{DS07, KS07, KS09, KS11, SS11, SS12},
set-multilinear circuits \cite{RS05, FS12, ASS13}, read-once oblivious ABP (ROABP) \cite{RS05, FS13,FSS14, AGKS14}. 
This lack of progress is not surprising:
Gupta~et~al.~\cite{GKKS13} showed
that a polynomial time test for \mbox{depth-$3$} circuits would imply a sub-exponential time 
test for general circuits. 
For now, even a sub-exponential solution for depth-$3$ circuits seems elusive. 
However, an efficient test for depth-$3$ multilinear circuits looks within reach 
as a lower bound against this class of circuits is already known~\cite{RY09}.
A circuit is called \emph{multilinear\/} if all its gates compute a multilinear polynomial,
i.e.\ polynomials such that the maximum degree of any variable is one.

A depth-$3$ multilinear circuit is called \emph{set-multilinear\/}
if all the product gates in it induce the same partition on the set of variables.
It is easy to see that a depth-$3$ multilinear circuit is a sum of polynomially many 
set-multilinear circuits.
Hence,
a natural first step to attack depth-$3$ multilinear circuit
is to find an efficient test for the sum of two set-multilinear polynomials. 
Before this work, the only non-trivial test known for sum of two set-multilinear circuits
was a sub-exponential whitebox
algorithm by Agrawal~et~al.~\cite{AGKS14}.
Subsequently, a sub-exponential time blackbox test was also given for 
depth-$3$ multilinear circuits \cite{OSV14}.
Our results imply the first polynomial-time whitebox algorithm,
and the first quasi-polynomial-time blackbox algorithm,
for the sum of two set-multilinear circuits.

In this paper,
we deal with ROABPs,
a model which subsumes set-multilinear circuits;
see for example~\cite[Lemma 14]{AGKS14}.
A read-once oblivious ABP (ROABP) is an arithmetic branching program,
where each variable occurs in at most one layer. 
There has been a long chain of work on identity testing for ROABP,
see the thesis of Michael Forbes~\cite{Forbes14} for an excellent overview.
In 2005, Raz and Shpilka~\cite{RS05} gave a polynomial-time whitebox
test for ROABP.
Then, Forbes and Shpilka~\cite{FS13} gave an $s^{O(\log n)}$-time blackbox algorithm
for ROABP with known variable order, where $s$ is the size of the ROABP and $n$ is number of variables.
This was followed by a complete blackbox test~\cite{FSS14} that took $s^{O(d\log^2 s)}$ steps,
where $d$ is the syntactic degree bound of any variable.
This was further improved by Agrawal~et~al.~\cite{AGKS14} to $s^{O(\log n)}$ time.
They removed the exponential dependence on the degree~$d$.
Their test is based on the idea of {\em basis isolating weight assignment\/}.
Given a polynomial over an algebra, it
assigns weights to the variables, and naturally extends it to monomials, 
such that there is a unique minimum weight basis
among the coefficients of the polynomial.

In another work, Jansen et al.~\cite{JQS10} gave a blackbox test for a sum of constantly many
``ROABPs''. Their definition of ``ROABP'' is much weaker. They 
assume that a variable appears on at most one edge in the ABP.

We consider the sum of ROABPs.
Note that
there are polynomials~$P(\x)$ computed by the sum of two ROABPs 
such that any single ROABP that computes~$P(\x)$ has exponential size~\cite{NS14}.
Hence, the previous results on single ROABPs do not help here.
In Section~\ref{sec:whitebox} we show our first main result (Theorem \ref{thm:PITcROABPwb}):
\begin{quote}
{\em
PIT for the sum of constantly many {\rm ROABP}s is in polynomial time.
}
\end{quote}
The exact time bound we get for the PIT-algorithm is $(nd w^{2^c})^{O(c)}$,
where~$n$ is the number of variables,
$d$ is the degree bound of the variables,
$c$ is the number of ROABPs and
$w$ is their width.
Hence
our time bound is double exponential in~$c$,
but polynomial in~$ndw$.

Our algorithm uses the fact that the {\em evaluation dimension\/} of an ROABP
is equal to the width of the ROABP~\cite{Nis91, FS13a}.
Namely,
we consider a set of linear dependencies derived from partial evaluations
of the ROABPs
\footnote{Equivalently, we work with the dependencies of the partial derivatives.}.
We view identity testing of the sum of two ROABPs 
as testing the equivalence of two ROABPs.
Our idea is inspired from a similar result in the boolean case. 
Testing the equivalence of two ordered boolean branching programs (OBDD)
is in polynomial time~\cite{SW97}.
OBDDs too have a similar property of small evaluation dimension, except that
the notion of {\em linear dependence\/} becomes {\em equality\/} in the boolean setting.
Our equivalence test, for two ROABPs $A$ and $B$, 
takes linear dependencies among partial evaluations of $A$
and verifies them for the corresponding partial evaluations of $B$. 
As $B$ is an ROABP, the verification of these dependencies 
reduces to identity testing for a single ROABP.

In Section~\ref{sec:sumOfC}, we generalize this test to the sum of~$c$ ROABPs. 
There we take~$A$ as one ROABP 
and~$B$ as the sum of the remaining~$c-1$ ROABPs. 
In this case, the verification of the dependencies for~$B$ becomes
the question of identity testing of a sum of~$c-1$ ROABPs,
which we solve recursively.

The same idea can be applied to decide the equivalence of an OBDD with
the XOR of~$c-1$ OBDDs. We skip these details here as we
are mainly interested in the arithmetic case.


In Section~\ref{sec:blackbox},
we give an identity test for a sum of ROABPs in the blackbox setting.
That is, 
we are given blackbox access to a sum of ROABPs and \emph{not\/} to the individual ROABPs. 
Our main result here is as follows (Theorem~\ref{thm:blackbox}):

\begin{quote}
{\em
There is a blackbox PIT for  the sum of constantly many {\rm ROABP}s 
that works in quasi-polynomial time.
}
\end{quote}
The exact time bound we get for the PIT-algorithm is $(ndw)^{O(c\, 2^c\log(ndw))}$,
where~$n$ is the number of variables,
$d$ is the degree bound of the variables,
$c$ is the number of ROABPs and
$w$ is their width.
Hence
our time bound is double exponential in~$c$,
and quasi-polynomial in~$n,d,w$.

Here again, using the low evaluation dimension property, the question
is reduced to identity testing for a single ROABP. 
But, just a hitting-set for ROABP does not suffice here, we need 
an efficient shift of the variables which gives low-support concentration 
in any polynomial computed by an ROABP.
An $\ell$-concentration in a polynomial~$P(\x)$ means that all of its coefficients
are in the linear span of its  coefficients corresponding to monomials with support~$< \ell$. 
Essentially we show that a shift, which achieves low-support concentration
for an ROABP of width~$w^{2^c}$, 
also works for a sum of~$c$  ROABPs (Lemma~\ref{lem:blackBoxSumOfC}).
This is surprising, because as mentioned above,
a sum of~$c$ ROABPs is not captured
by an ROABP with polynomially bounded width~\cite{NS14}.

A novel part of our proof is the idea that 
for a polynomial over a $k$-dimensional $\F$-algebra~$\A_k$,
a shift by a basis isolating weight assignment achieves low-support concentration.
To elaborate, let $\w \colon \x \to \N$ be a basis isolating weight assignment
for a polynomial $P(\x) \in \A_k[\x]$ then $P(\x + t^\w)$ has $O(\log k)$-concentration
over~$\F(t)$.
As Agrawal~et~al.~\cite{AGKS14} gave a basis isolating weight assignment for ROABPs, we can use it to 
get low-support concentration.
Forbes et al.~\cite{FSS14} had also achieved low-support concentration in ROABPs,
but with a higher cost.
Our concentration proof significantly differs from the older 
rank concentration proofs~\cite{ASS13,FSS14}, which
always assume \emph{distinct\/} weights for all the monomials or coefficients. 
Here, we only require that the weight of a coefficient is greater than 
the weight of the basis coefficients that it depends on.


\section{Preliminaries}
\label{sec:prelim}

\subsection{Notation}
Let $\x = (x_1,x_2, \dots, x_n)$ be a tuple of $n$ variables. 
For any $\a = (a_1, a_2, \dots, a_n ) \in \N^n$, 
we denote by $\xa$  the monomial $\prod_{i=1}^n x_i^{a_i}$. 
The \emph{support size\/} of a monomial~$\xa$ is given by 
$\supp(\a) = \abs{\{ a_i \neq 0 \mid i \in [n] \}}$.

Let $\F$ be some field.
Let~$A(\x)$ be a polynomial over~$\F$ in~$n$ variables.
A polynomial $A(\x)$ is said to have \emph{individual degree} $d$, 
if the degree of each variable is bounded by $d$ for each monomial in $A(\x)$.
When~$A(\x)$ has individual degree~$d$, then the exponent~$\a$ of any monomial~$\xa$ of~$A(\x)$
is in the set
\[
 \M = \{0,1, \dots, d\}^n \,.
\]
By~$\coeff_{A}(\xa) \in \F$ we denote the coefficient of the monomial~$\xa$ in~$A(\x)$.
Hence, we can write
\[
 A(\x) = \sum_{\a \in \M} \coeff_{A}(\xa)\, \xa \,.
\]
The {\em sparsity\/} of polynomial~$A(\x)$ is the number of nonzero coefficients~$\coeff_{A}(\xa)$.

We also consider {\em matrix polynomials\/}
where the coefficients~$\coeff_{A}(\xa)$ are $w \times w$ matrices, for some~$w$.
In an abstract setting,
these are polynomials over a $w^2$-dimensional $\F$-algebra~$\A$.
Recall that an $\F$-algebra is a vector space over~$\F$ with a
multiplication which is bilinear and associative,
i.e.\ $\A$ is a ring.
The {\em coefficient space\/} is then defined as the span of all coefficients
of~$A$,
i.e., 
$\Span_{\F}\{\coeff_A(\x^a) \mid \a \in \M \}$.

Consider a partition of the variables~$\x$ into two parts~$\y$ and~$\z$, with $\abs{\y}=k$.
A polynomial~$A(\x)$ can be viewed as a polynomial in  variables~$\y$, 
where the coefficients are polynomials in~$\F[\z]$.
For monomial~$\ya$,
let us denote the coefficient of~$\ya$ in~$A(\x)$ by $\coeffset{A}{\y}{\a} \in \F[\z]$.
For example, in the polynomial $A(\x) = x_1 + x_1x_2 + {x_1}^2$, we have
$\coeffset{A}{x_1}{1} = 1 + x_2$, whereas
$\coeff_{A}(x_1) = 1$.
Observe that $\coeff_{A}(\ya)$ is the constant term in $\coeffset{A}{\y}{\a}$.

Thus, $A(\x)$ can be written as 
\begin{equation}\label{eq:C_ya}
A(\x) = \sum_{\a \in \{0,1, \dots,d\}^k} \coeffset{A}{\y}{\a} \, \ya \,.
\end{equation}
The coefficient $\coeffset{A}{\y}{\a}$ is also sometimes expressed 
in the literature as a partial derivative~$\frac{\partial A}{\partial \ya} $ 
evaluated at $\y = {\boldsymbol 0}$
(and multiplied by an appropriate constant), see~\cite[Section 6]{FS13a}.

For a set of polynomials~$\mathcal{P}$,
we define their $\F$-$\Span$ as
\[
\Span_{\F} \mathcal{P} = \cbrace{\sum_{A \in \mathcal{P}}
\alpha_A A \mid \alpha_A \in \F \text{ for all } A \in \mathcal{P}}.
\]
The set of polynomials~$\mathcal{P}$ is said to be $\F$-{\em linearly independent\/}
if $\sum_{A \in \mathcal{P}} \alpha_A A = 0$ holds only for $\alpha_A = 0$, 
for all~$A \in \mathcal{P}$.
The \emph{dimension\/}~$\dim_\F \mathcal{P}$ of~$\mathcal{P}$ is
the cardinality of the largest $\F$-linearly independent subset of~$\mathcal{P}$.

For a matrix~$R$, we denote by~$R(i,\mydot)$ and~$R(\mydot, i)$ 
the $i$-th row and the $i$-th column of~$R$, respectively.
For any $a \in \F^{k \times k'}, b \in \F^{\ell \times \ell'}$, 
the tensor product of~$a$ and~$b$ is denoted by $a\otimes b$.
The inner product is denoted by~$\abrace{a,b}$.
We abuse this notation slightly: 
for any $a, R \in \F^{w \times w}$,
let $\abrace{a,R} = \sum_{i = 1}^w\sum_{j = 1}^w a_{i j} R_{i j}$. 


\subsection{Arithmetic branching programs}
\label{sec:abp}
An {\em arithmetic branching program\/} (ABP) is a directed graph 
with $\ell+1$ layers of vertices $(V_0,V_1, \dots, V_{\ell})$.
The layers $V_0$ and $V_\ell$ each contain only one vertex, 
the {\em start node\/}~$v_{0}$ and 
the {\em end node\/}~$v_{\ell}$, respectively.
The edges are 
only going from the vertices in the layer $V_{i-1}$ to the vertices in the layer $V_i$, 
for any $i \in [d]$.
All the edges in the graph have weights from~$\F[\x]$,
for some field~$\F$. 
The {\em length\/} of an ABP is the length of a longest path in the ABP, i.e.~$\ell$.
An ABP has {\em width\/}~$w$, 
if $\abs{V_i} \leq w$ for all $1 \leq i \leq \ell-1$.

For an edge~$e$, let us denote its weight by~$W(e)$.
For a path~$p$,
its weight~$W(p)$ is defined to be the product of weights of all the edges
in it,
\[ W(p) = \prod_{e \in p} W(e).\]
The {\em polynomial $A(\x)$ computed by the ABP\/}
is the sum of the weights of all the paths from $v_{0}$ to $v_{\ell}$,
\[
A(\x) = \sum_{p \text{ path } v_{0} \leadsto v_{\ell}} W(p).
\]

Let the set of nodes in $V_i$ be $\{v_{i,j} \mid j \in [w]\}$.
The branching program can alternately be represented by a matrix product
$\prod_{i=1}^{\ell} D_i  $,
where $D_1 \in \F[\x]^{1\times w}$, 
$D_i \in \F[\x]^{w \times w}$ for $2 \leq i \leq \ell-1$,
and $D_{\ell} \in \F[\x]^{w\times 1}$ 
such that 
\begin{eqnarray*}
D_1(j) &=& W(v_0,v_{1,j}),\; \text{ for } 1 \leq j \leq w,\\
D_i(j, k) &=& W(v_{i-1,j},v_{i,k}),\;  \text{ for } 1 \leq j,k \leq w \text{ and } 2 \leq i \leq n-1,\\
D_{\ell}(k) &=& W(v_{\ell-1,k},v_{\ell}),\;  \text{ for } 1 \leq k \leq w.
\end{eqnarray*}
Here we use the convention that $W(u,v) = 0$ if $(u,v)$ is not an edge in the ABP.


\subsection{Read-once oblivious arithmetic branching programs}
\label{subsec:roabpCharacterization}

An ABP is called a {\em read-once oblivious ABP (ROABP)}
if the edge weights in every layer are univariate polynomials in the same variable,
and every variable occurs in at most one layer.
Hence,
the length of an ROABP is~$n$, the number of variables.
The entries in the matrix~$D_{i}$ defined above come from~$\F[x_{\pi(i)}]$,
for all $i \in [n]$, where $\pi$ is a permutation on the set~$[n]$.
The order $(x_{\pi(1)}, x_{\pi(2)}, \dots, x_{\pi(n)})$ is said 
to be the \emph{variable order\/} of the ROABP.

We will view~$D_{i}$ as a polynomial in the variable~$x_{\pi(i)}$,
whose coefficients are $w$-dimensional vectors or matrices.
Namely,
for an exponent $\a = (a_1, a_2, \dots, a_n)$,
the coefficient of
\begin{itemize}
\item
$x_{\pi(1)}^{a_{\pi(1)}}$ in~$D_1(x_{\pi(1)})$ is the row vector
$\coeff_{D_1}(x_{\pi(1)}^{a_{\pi(1)}}) \in \F^{1 \times w}$,
\item 
$x_{\pi(i)}^{a_{\pi(i)}}$ in~$D_i(x_{\pi(i)})$ is the matrix
$\coeff_{D_i}(x_{\pi(i)}^{a_{\pi(i)}}) \in \F^{w \times w}$, for $i = 2,3, \dots, n-1$, and
\item
$x_{\pi(n)}^{a_{\pi(n)}}$ in~$D_n(x_{\pi(n)})$ is the vector
$\coeff_{D_n}(x_{\pi(n)}^{a_{\pi(n)}}) \in \F^{w \times 1}$.
\end{itemize}

The read once property gives us an easy way to express the coefficients of the
polynomial~$A(\x)$ computed by an ROABP.
\begin{lemma}
\label{lem:coeffprod}
For a polynomial $A(\x) = D_1(x_{\pi(1)}) D_2(x_{\pi(2)}) \cdots D_n(x_{\pi(n)})$ computed by an {\rm ROABP},
we have 
\begin{equation}
 \coeff_{A}(\xa) = \prod_{i=1}^n \coeff_{D_i}(x_{\pi(i)}^{a_{\pi(i)}})~~ \in \F \,.
 \label{eq:coeff}
\end{equation}
\end{lemma}

We also consider matrix polynomials computed by an ROABP.
A matrix polynomial $A(\x) \in F^{w \times w}[\x]$ is said to be computed by an ROABP if
$A = D_1 D_2 \cdots D_n$, where $D_i \in F^{w \times w}[x_{\pi(i)}]$
for $i = 1, 2, \dots, n$ and some permutation~$\pi$ on~$[n]$.
Similarly, a vector polynomial $A(\x) \in F^{1 \times w}[\x]$
is said to be computed by an ROABP if 
$A = D_1 D_2 \cdots D_n$, where $D_1 \in F^{1 \times w}[x_{\pi(1)}]$ 
and $D_i \in F^{w \times w}[x_{\pi(i)}]$ for $i =  2, \dots, n$.
Usually, we will assume that an ROABP computes a polynomial in $\F[\x]$,
unless mentioned otherwise.

Let~$A(\x)$ be the polynomial computed by an ROABP and
let~$\y$ and~$\z$ be a partition of the variables~$\x$ 
such that~$\y$ is a \emph{prefix} of the variable order of the ROABP.
Recall from equation~(\ref{eq:C_ya}) that 
$\coeffset{A}{\y}{\a} \in \F[\z]$ is the coefficient of monomial~$\ya$ in~$A(\x)$.
Nisan~\cite{Nis91} showed that for every prefix~$\y$, 
the dimension of the set of coefficient polynomials~$\coeffset{A}{\y}{\a}$ 
is bounded by the width of the 
ROABP\footnote{Nisan~\cite{Nis91} showed it for non-commutative ABP, but the same proof works for ROABP.}.
This holds
in spite of the fact that the number of these polynomials is large.

\begin{lemma}[\cite{Nis91}, Prefix $\y$]
\label{lem:ROABPdim}
Let $A(\x)$ be a polynomial of individual degree~$d$, computed by an ROABP of width~$w$  with
variable order $(x_1, x_2, \dots, x_n)$. 
Let $k \leq n$ and $\y = (x_1, x_2, \dots, x_k)$ be the prefix of length~$k$ of~$\x$.
Then
$\dim_\F \{ \coeffset{A}{\y}{\a} \mid \a \in \{0,1,\dots, d\}^k \} \leq w.$
\end{lemma}

\begin{proof}
Let $A(\x) = D_1(x_1)\, D_2(x_2)\, \cdots\, D_n(x_n)$, 
where $D_1 \in \F^{1 \times w}[x_1]$, $D_n \in \F^{w \times 1}[x_n]$
and $D_i \in \F^{w \times w}[x_i]$, for $2 \leq i \leq n-1$.
Let $\z = (x_{k+1}, x_{k+2}, \dots, x_n)$ be the remaining variables of~$\x$. 
Define $P(\y) = D_1 D_2 \cdots D_k$ and $Q(\z) = D_{k+1} D_{k+2} \cdots D_n$.
Then~$P$ and~$Q$ are vectors of length~$w$,
\begin{align*}
 P(\y) &= [P_1(\y) \; P_2(\y) \; \cdots \; P_w(\y)]\\
Q(\z) &= [Q_1(\z) \; Q_2(\z) \; \cdots \; Q_w(\z)]^T
\end{align*}
where $P_i(\y) \in \F[\y]$ and $Q_i(\z) \in \F[\z]$, for $1 \leq i \leq w$,
and we have  $A(\x) = P(\y)\, Q(\z)$. 

We get the following generalization of equation~(\ref{eq:coeff}):
for any $\a \in \{0,1, \dots, d\}^k$, 
the coefficient~$\coeffset{A}{\y}{\a} \in \F[\z]$  of monomial~$\ya$
can be written as
\begin{equation}
 \coeffset{A}{\y}{\a} = \sum_{i=1}^w \coeff_{P_i}(\ya) \, Q_i(\z).
 \label{eq:coeffset}
\end{equation}
That is,
every~$\coeffset{A}{\y}{\a}$ is in the $\F$-span 
of the polynomials $Q_1, Q_2, \dots,Q_w$.
Hence,
the claim follows.
\end{proof}

Observe that equation~(\ref{eq:coeffset}) tells us that
the polynomials $\coeffset{A}{\y}{\a}$ can also be computed by an ROABP of width~$w$:
by equation~(\ref{eq:coeff}),
we have
$\coeff_{P_i}(\ya) = \prod_{x_i \in \y} \coeff_{D_i}(x_i^{a_i})$.
Hence,
in the ROABP for~$A$ we simply have to replace the matrices~$D_i$ which belong to~$P$
by the coefficient matrices~$\coeff_{D_i}(x_i^{a_i})$.
Here,
$\y$ is a prefix of~$\x$.
But this is not necessary for the construction to work.
The variables in~$\y$ can be arbitrarily distributed in~$\x$.
We summarize the observation in the following lemma.

\begin{lemma}[Arbitrary $\y$]
\label{lem:pdROABP}
Let $A(\x)$ be a polynomial of individual degree~$d$, computed by an ROABP of width~$w$
and $\y = (x_{i_1}, x_{i_2}, \dots, x_{i_k})$ be any $k$ variables of~$x$.
Then the polynomial $\coeffset{A}{\y}{\a}$ can be computed by an ROABP of width~$w$,
for every $\a \in \{0,1,\dots,d\}^k$. 
Moreover,
all these ROABPs have the same variable order,
inherited from the order of the ROABP for~$A$.
\end{lemma}

For a general polynomial, 
the dimension considered in Lemma~\ref{lem:ROABPdim} can be exponentially large in~$n$.
We will next show the converse of Lemma~\ref{lem:ROABPdim}:
if this dimension is small for a polynomial then
there exists a small width ROABP for that polynomial.
Hence, this property characterizes the class of polynomials
computed by ROABPs.
Forbes et al.~\cite[Section~6]{FS13a} give a similar characterization in terms of 
evaluation dimension, for polynomials which can be 
computed by an ROABP, in any variable order. 
On the other hand, we work with a fixed variable order.

As a preparation to prove this characterization
we define a {\em characterizing set of dependencies\/}
of a polynomial~$A(\x)$ of individual degree~$d$, with respect to a variable order $(x_1,x_2, \dots, x_n)$. 
This set of dependencies will essentially give us an ROABP for~$A$ in
the variable order $(x_1,x_2, \dots, x_n)$.

\begin{definition}\label{def:dependencies}
Let $A(\x)$ be polynomial of individual degree~$d$, where $\x = (x_1, x_2, \dots, x_n)$.
For any $0 \leq k \leq n$ and $\y_k = (x_1,x_2, \dots, x_k)$, 
let
$$\dim_\F \{ \coeffset{A}{\y_k}{\a} \mid \a \in \{0,1,\dots, d\}^k \} \leq w ,$$
for some~$w$.

For $0 \leq k \leq n$,
we define the {\em spanning sets\/}
$\spanning_k(A)$ and the {\em dependency sets\/} $\depending_k(A)$ as subsets of $\{ 0,1, \dots, d \}^k$ as follows.

For $k = 0$, let
$\depending_0(A) = \emptyset$ and
$\spanning_0(A) = \{ \epsilon \}$,
where $\epsilon = (\,)$ denotes the empty tuple. 
For $k > 0$, let
\begin{itemize}
\item 
$\depending_k(A) = \{ (\a , j) \mid \a \in \spanning_{k-1}(A) \text{ and } 0 \leq j \leq d \}$,
i.e.\
$\depending_k(A)$ contains all possible extensions of the tuples in $\spanning_{k-1}(A)$.
\item 
$\spanning_k(A) \subseteq \depending_k(A)$
is any set of size $\le w$,
such that 
for any $\b \in \depending_k(A)$, the polynomial~$\coeffset{A}{\y_k}{\b}$
is in the span of $\{\coeffset{A}{\y_k}{\a} \mid \a \in \spanning_k(A)\}$.
\end{itemize}
The dependencies of the polynomials in $\{\coeffset{A}{\y_k}{\a} \mid \a \in \depending_k(A)\}$
over $\{\coeffset{A}{\y_k}{\a} \mid \a \in \spanning_k(A)\}$ are the 
{\em characterizing set of dependencies}.
\end{definition}
The definition of $\spanning_k(A)$ is not unique.
For our purpose, it does not matter which of the possibilities  we take,
we simply fix one of them.
We do {\em not\/} require that $\spanning_k(A)$ is of minimal size,
i.e.\ the polynomials associated with $\spanning_k(A)$ constitute a basis
for the  polynomials associated with $\depending_k(A)$.
This is because in the whitebox test in Section~\ref{sec:whitebox},
we will efficiently construct the sets $\spanning_k(A)$,
and there we cannot guarantee to obtain a basis.
We will see that it suffices to have $|\spanning_k(A)| \leq w$.
It follows that $|\depending_{k+1}(A)| \leq w(d+1)$.
Note that for $k = n$,
we have $\y_n = \x$ and therefore $\coeffset{A}{\y_n}{\a} = \coeff_A(\xa)$
is a constant for every~$\a$.
Hence,
the coefficient space has dimension one in this case,
and thus $|\spanning_n(A)| = 1$.

Now we are ready to construct an ROABP for $A$.


\begin{lemma}[\cite{Nis91}, Converse of Lemma \ref{lem:ROABPdim}]
Let $A(\x)$ be a polynomial of individual degree~$d$
with $\x = (x_1, x_2, \dots, x_n)$,
such that for any $1 \leq k \leq n$ and $\y_k = (x_1,x_2, \dots, x_k)$, 
we have
\[
\dim_\F \{\, \coeffset{A}{\y_k}{\a} \mid \a \in \{0,1,\dots, d\}^k\, \} \leq w \,.
\]
Then there exists an  {\rm ROABP} of width~$w$ for~$A(\x)$ in the variable order
$(x_1, x_2, \dots, x_n)$.
\label{lem:dimROABP}
\end{lemma}

\begin{proof}
To keep the notation simple,
we assume\footnote{Assuming $d+1 \ge w,~\spanning_k(A)$ can be made to have size $=w$ for each $k$.}
that $|\spanning_k(A)| = w$
for each $1 \leq k \leq n-1$.
The argument would go through even when $|\spanning_k(A)| < w$.
Let  $\spanning_k(A) = \{ \a_{k,1},\a_{k,2}, \dots, \a_{k,w} \}$
and $\spanning_n(A) = \{ \a_{n,1} \}$.

To prove the claim,
we construct matrices $D_1,D_2, \dots, D_n$, 
where $D_1 \in \F[x_1]^{1 \times w}$, $D_n \in \F[x_n]^{w \times 1}$,
and $D_i \in \F[x_i]^{w \times w}$, for $i = 2, \dots, n-1$,
such that $A(\x) = D_1\, D_2 \cdots D_n$.
This representation shows that there is an ROABP of width~$w$ for~$A(\x)$.

The matrices are constructed inductively
such that for $k = 1,2 \dots, n-1$,
\begin{equation}
A(\x) = D_1 D_2 \cdots D_k 
\, [\coeffset{A}{\y_k}{\a_{k,1}}\; \coeffset{A}{\y_k}{\a_{k,2}}\;  \cdots \; \coeffset{A}{\y_k}{\a_{k,w}} ]^T \,.
\label{eq:D1-k}
\end{equation}

To construct $D_1 \in \F[x_1]^{1 \times w}$, consider the equation
\begin{equation}
A(\x) = \sum_{j=0}^d \coeffset{A}{\y_1}{j}\, x_1^j.
\label{eq:Apart1}
\end{equation}
Recall that $\depending_1(A) = \{0,1, \dots, d\}$.
By the definition of $\spanning_1(A)$,
every $\coeffset{A}{\y_1}{j}$ is in the span of the $\coeffset{A}{\y_1}{\a}$'s
 for $\a \in \spanning_1(A)$.
That is,
there exists constants~$\{\gamma_{j,i}\}_{i,j}$
such that for all $0 \leq j \leq d$ we have
\begin{equation}
\coeffset{A}{\y_1}{j} = \sum_{i=1}^{w} \gamma_{j,i}\, \coeffset{A}{\y_1}{\a_{1,i}}. 
\label{eq:Agamma1}
\end{equation}
From equations~(\ref{eq:Apart1}) and~(\ref{eq:Agamma1}) 
we get, 
$A(\x) = \sum_{i=1}^{w} \left( \sum_{j=0}^d \gamma_{j,i}\, x_1^j \right) \coeffset{A}{\y_1}{\a_{1,i}}.$
Hence,
we define $D_1 = [D_{1,1} \; D_{1,2} \; \cdots \; D_{1,{w}}]$,
where $D_{1,i} = \sum_{j=0}^d \gamma_{j,i}\, x_1^j$, for all $i \in [{w}]$.
Then we have
\begin{equation}
A = D_1\, [\coeffset{A}{\y_1}{\a_{1,1}} \; \coeffset{A}{\y_1}{\a_{1,2}} \; \cdots \; \coeffset{A}{\y_1}{\a_{1,{w}}} ]^T.
\label{eq:layer1}
\end{equation}

To construct $D_{k} \in \F[x_{k}]^{w \times w}$ for $2 \leq k \leq n-1$,
we consider the equation
\begin{equation}
[\coeffset{A}{\y_{k-1}}{\a_{k-1,1}} \cdots \coeffset{A}{\y_{k-1}}{\a_{k-1,w}} ]^T
= D_{k} \,
[\coeffset{A}{\y_{k}}{\a_{{k},1}} \cdots \coeffset{A}{\y_{k}}{\a_{{k},w}} ]^T \,.
\label{eq:layerk}
\end{equation}

We know that for each $1 \leq i \leq w$,
\begin{equation}
\coeffset{A}{\y_{k-1}}{\a_{k-1,i}} = \sum_{j=0}^d \coeffset{A}{\y_{k}}{ (\a_{k-1,i}, j) }\, x_{k}^j.
\label{eq:Apartk}
\end{equation}
Observe that~$(\a_{k-1,i}, j)$ is just an extension of~$\a_{k-1,i}$ and
thus belongs to $\depending_{k}(A)$.
Hence, there exists a set of constants
$\{\gamma_{i,j,h}\}_{i,j,h}$
such that for all $0 \leq j \leq d$ we have
\begin{equation}
\coeffset{A}{\y_{k}}{ (\a_{k-1,i}, j)} = \sum_{h=1}^{w} \gamma_{i,j,h}\, \coeffset{A}{\y_{k}}{\a_{{k},h}} . 
\label{eq:Agammak}
\end{equation}

From equations~(\ref{eq:Apartk}) and~(\ref{eq:Agammak}), 
for each $1 \leq i \leq w$ we get
\[
\coeffset{A}{\y_{k-1}}{\a_{k-1,i}} = \sum_{h=1}^{w} \left( \sum_{j=0}^d \gamma_{i,j,h}\, x_{k}^j \right)
\coeffset{A}{\y_{k}}{\a_{{k},h}}\,.
\]
Hence, we can define $D_{k}(i,h) = \sum_{j=0}^d \gamma_{i,j,h}\, x_{k}^j$,
for all $i,h \in [w]$.
Then~$D_{k}$ is the desired matrix in equation~(\ref{eq:layerk}).

Finally,
we obtain~$D_n \in \F^{w \times 1}[x_n]$ in an analogous way.
Instead of equation~(\ref{eq:layerk}) we consider the equation
\begin{equation}
[\coeffset{A}{\y_{n-1}}{\a_{n-1,1}} \cdots \coeffset{A}{\y_{n-1}}{\a_{n-1,w}} ]^T
= D'_{n} \,
[\coeffset{A}{\y_n}{\a_{n,1}}] \,.
\label{eq:layern}
\end{equation}
Recall that~$\coeffset{A}{\y_n}{\a_{n,1}} \in \F$ is a constant
that can be absorbed into the last matrix~$D'_n$,
i.e.\ we define $D_n = D'_n\, \coeffset{A}{\y_n}{\a_{n,1}}$.
Combining equations~(\ref{eq:layer1}),~(\ref{eq:layerk}), and~(\ref{eq:layern}),
we get $A(\x) = D_1\, D_2 \cdots D_n$.
\end{proof}

Consider the polynomial~$P_k$ defined as the product of the first~$k$ matrices
$D_1, D_2, \dots, D_k$ from the above proof;
$P_k(\y_k) = D_1 D_2 \cdots D_k$.
We can write~$P_k$ as 
\[
P_k(\y_k) = \sum_{\a \in \{0,1, \dots,d\}^k} \coeff_{P_k}(\y_k^{\a})\, \y_k^{\a}\,,
\]
where $\coeff_{P_k}(\y_k^{\a})$ is a vector in~$\F^{1 \times w}$.
We will see next
that it follows from the proof of Lemma~\ref{lem:dimROABP}  that the
coefficient space of~$P_k$,
i.e., 
$\Span_{\F} \{ \coeff_{P_k}(\y_k^{\a}) \mid \a \in \{0,1, \dots,d\}^k \}$ 
has full rank~$w$.

\begin{corollary}[Full Rank Coefficient Space]
\label{cor:coeffspaceP}
Let $D_1, D_2, \dots ,D_n$ be the matrices constructed in the proof of Lemma~\ref{lem:dimROABP} 
with $A = D_1 D_2 \cdots D_n$.
Let $\spanning_k(A) = \{ \a_{k,1},\a_{k,2}, \dots, \a_{k,w} \}$.
For $k \in [n]$, define the polynomial $P_k(\y_k) = D_1 D_2 \cdots D_k$.

Then for any $\ell \in [w]$, we have
$\coeff_{P_k}(\y_k^{\a_{k,\ell}}) = \e_{\ell}$,
where~$\e_{\ell}$ is
the $\ell$-th {\em elementary unit vector\/},
$\e_{\ell} = (0, \dots, 0,1,0, \dots, 0)$ of length~$w$, with a one at position~$\ell$,
and zero at all other positions.
Hence, the coefficient space of~$P_k$ has full rank~$w$.
\end{corollary}

\begin{proof}
In the construction of the matrices~$D_k$ in the proof of Lemma~\ref{lem:dimROABP},
consider the special case in equations~(\ref{eq:Agamma1}) and~(\ref{eq:Agammak})
that the  exponent~$(\a_{k-1,i}, j)$ is in~$\spanning_k(A)$,
say $(\a_{k-1,i}, j)= \a_{k,\ell} \in \spanning_k(A)$.
Then the $\gamma$-vector to express~$\coeffset{A}{\y_{k}}{ (\a_{k-1,i}, j)}$
in equation~(\ref{eq:Agamma1}) and~(\ref{eq:Agammak}) can be chosen to be~$\e_{\ell}$,
i.e.\ 
$\left(\gamma_{i,j,h}\right)_h = \e_{\ell}$.
By the definition of matrix~$D_k$, vector~$\e_{\ell}$ becomes the $i$-th row
of~$D_k$ for the exponent~$j$,
i.e.,
$\coeff_{\row{D_k}{i}}(x_k^{j})=\e_{\ell}$.

This shows the claim for $k=1$.
For larger~$k$, it follows by induction because for $(\a_{k-1,i}, j)= \a_{k,\ell}$
we have
$\coeff_{P_k}(\y_k^{\a_{k,\ell}}) = \coeff_{P_{k-1}}(\y_{k-1}^{\a_{k-1,i}}) \coeff_{D_k}(x_k^{j})$ \,.
\end{proof}


\section{Whitebox Identity Testing}
\label{sec:whitebox}

We will use the characterization of ROABPs provided by 
Lemmas~\ref{lem:ROABPdim} and~\ref{lem:dimROABP} in Section~\ref{sec:twoROABP}
to design a polynomial-time algorithm to check if two given ROABPs are equivalent.
This is the same problem as checking whether the sum of two ROABPs is zero.
In Section~\ref{sec:sumOfC},
we extend the test to check whether the sum of constantly many ROABPs is zero.


\subsection{Equivalence of two ROABPs}
\label{sec:twoROABP}

Let $A(\x)$ and $B(\x)$ be two polynomials of individual degree~$d$,
given by two ROABPs.
If the two ROABPs have the same variable order then
one can combine them into a single ROABP which computes their difference.
Then one can apply the test for one ROABP 
(whitebox \cite{RS05}, blackbox \cite{AGKS14}). 
So, the problem is non-trivial only when the two ROABPs have different variable order.
W.l.o.g.\ we assume that~$A$ has order $(\lis{x}{,}{n})$.
Let~$w$ bound the width of both ROABPs.
In this section we prove that we can find out in polynomial time 
whether $A(\x) = B(\x)$.

\begin{theorem}\label{thm:2ROABPwb}
The equivalence of two ROABPs can be tested in polynomial time.
\end{theorem}

The idea is to determine the characterizing set of dependencies
among the partial derivative polynomials of~$A$,
and verify that the same dependencies hold for the corresponding partial derivative polynomials of~$B$.
By Lemma~\ref{lem:dimROABP},
these dependencies essentially define an ROABP.
Hence,
our algorithm is to construct an ROABP for~$B$ in the variable order of~$A$.
Then it suffices to check whether we get the same ROABP,
that is, whether all the matrices $D_1, D_2, \dots, D_n$ constructed
in the proof of Lemma~\ref{lem:dimROABP}
are the same for~$A$ and~$B$.
We give some more details.

\paragraph*{Construction of $\spanning_k(A)$.}
Let $A(\x) = D_1(x_1) D_2(x_2) \cdots D_n(x_n)$ of width~$w$.
We give an iterative construction, starting from $\spanning_0(A) = \{\epsilon\}$.
Let $1 \leq k \leq n$.
By definition, $\depending_k(A)$ consists of all possible one-step extensions of $\spanning_{k-1}(A)$.
Let
$\b = (b_1, b_2, \dots, b_k) \in \{0,1, \dots, d\}^k$.
Define
\[
C_{\b} = \prod_{i=1}^k \coeff_{D_i}(x_i^{b_i})\,.
\]
Recall that
$\coeff_{D_1}(x_1^{b_1}) \in \F^{1 \times w}$ and
$\coeff_{D_i}(x_i^{b_i}) \in \F^{w \times w}$, for  $2 \leq i \leq k$.
Therefore  $C_{\b} \in \F^{1 \times w}$ for $k<n$.
Since $D_n \in \F^{w \times 1}$, we have $C_{\b} \in \F$ for $k=n$.
By equation~(\ref{eq:coeffset}), we have
\begin{equation}
\coeffset{A}{\y_k}{\b} = C_{\b}\, D_{k+1} \cdots D_n \,. 
\label{eq:coeffAyb}
\end{equation}
Consider the set of vectors ${\cal D}_k = \{ C_{\b} \mid \b \in \depending_k(A)\}$.
This set has dimension bounded by~$w$ since the width of~$A$ is~$w$.
Hence, we can determine a set ${\cal S}_k \subseteq {\cal D}_k$
of size~$\le w$ such that~${\cal S}_k$ spans~${\cal D}_k$.
Thus we can take
$\spanning_k(A)  = \{ \a \mid C_{\a} \in {\cal S}_k \}$.
Then,
for any $\b \in \depending_k(A)$,
vector~$C_{\b}$ is a linear combination
\[
C_\b = \sum_{\a \in \spanning_k(A)} \gamma_{\a}\, C_\a \,.
\]
Recall that $|\depending_k(A)| \leq w(d+1)$, i.e.\ this is a small set.
Therefore, we can efficiently compute the coefficients~$\gamma_{\a}$
for every $\b \in \depending_k(A)$ .
Note that by equation~(\ref{eq:coeffAyb}) we have the same dependencies
for the polynomials~$\coeffset{A}{\y_k}{\b}$.
That is, with the same coefficients~$\gamma_{\a}$,
we can write
\begin{equation}
\coeffset{A}{\y_k}{\b} = \sum_{\a \in \spanning_k(A)} \gamma_\a\, \coeffset{A}{\y_k}{\a} \,.
\label{eq:depA}
\end{equation}

\paragraph*{Verifying the dependencies for $B$.}
We want to verify that the dependencies in equation~(\ref{eq:depA})
computed for~$A$ hold for~$B$ as well,
i.e. that for all $k \in [n]$ and $\b \in \depending_k(A)$,
\begin{equation}
\coeffset{B}{\y_k}{\b} = \sum_{\a \in \spanning_k(A)} \gamma_\a\, \coeffset{B}{\y_k}{\a} \,.
\label{eq:depB}
\end{equation}

Recall that $\y_k = (x_1, x_2, \dots, x_k)$ and the ROABP for~$B$
has a different variable order.
By Lemma~\ref{lem:pdROABP},
every polynomial~$\coeffset{B}{\y_k}{\a}$
has an ROABP of width~$w$ and the same order on the
remaining variables as the one given for~$B$.
It follows that each of the $w+1$ polynomials that occur in equation~(\ref{eq:depB})
has an ROABP of width~$w$ and the same variable order.
Hence, we can construct {\em one\/} ROABP for the polynomial
\begin{equation}
\coeffset{B}{\y_k}{\b} - \sum_{\a \in \spanning_k(A)} \gamma_\a \coeffset{B}{\y_k}{\a} \,.
\label{eq:depB'}
\end{equation}
Simply identify all the start nodes and all the end nodes
and put the appropriate constants~$\gamma_{\a}$ to the weights.
Then we get an ROABP of width~$w(w+1)$.
In order to verify equation~(\ref{eq:depB}),
it suffices to do a zero-test for this ROABP.
This can be done in polynomial time~\cite{RS05}.

\paragraph*{Constructing ROABP for $B$ in the same sequence as $A$}
Recall Lemma~\ref{lem:dimROABP} and its proof.
There, we constructed an ROABP just from the characterizing dependencies
of the given polynomial.
Hence,
the construction applied to~$B$ will give an ROABP of width~$w$ for~$B$ 
with the same variable order $(x_1, x_2, \dots, x_n)$ as for~$A$.
The matrices~$D_k$ will be the same as those for~$A$ because their definition
uses only the dependencies provided by equation~(\ref{eq:depB}),
and they are the same as those for~$A$ in equation~(\ref{eq:depA}).

The last matrix $D_n$ can be written as $D_n' \coeffset{A}{\y_{n}}{\a_{n,1}}$,
similar to equation~(\ref{eq:layern}).
Since 
the dependencies of the coefficients in $\depending_n(B)$ over
coefficents in $\spanning_n(B)$
are the same as those for $A$,
$B(\x) = D_1 D_2 \cdots D_n'\, \coeffset{B}{\y_{n}}{\a_{n,1}}$.

\paragraph*{Checking Equality.}
Clearly,
if equation~(\ref{eq:depB}) fails to hold for some~$k$ and~$\b$,
then $A \not= B$.
When equation~(\ref{eq:depB}) holds for all~$k$ and~$\b$,
we only need to check if $\coeffset{A}{\y_{n}}{\a_{n,1}} = \coeffset{B}{\y_{n}}{\a_{n,1}}$,
which is a single evaluation of each ROABP.
This proves Theorem~\ref{thm:2ROABPwb}.



\subsection{Sum of constantly many ROABPs}
\label{sec:sumOfC}

Let $A_1(\x), A_2(\x), \dots, A_c(\x)$ be polynomials of individual degree~$d$,
given by~$c$ ROABPs.
Our goal is to test whether $A_1 + A_2 + \cdots + A_c = 0.$
Here again, the question is interesting only when the ROABPs have different variable orders.
We show how to reduce the problem to the case of the equivalence of two ROABPs
from the previous section.
For constant~$c$ this will lead to a polynomial-time test.

We start by rephrasing the  problem as an equivalence test.
Let $A = -A_1$ and $B = A_2 + A_3 + \cdots + A_c$. 
Then the problem has become to check whether $A = B$.
Since~$A$ is computed by a single ROABP,
we can use the same approach as in Section~\ref{sec:twoROABP}.
Hence, we again get the dependencies from equation~(\ref{eq:depA}) for~$A$.
Next,
we have to verify these dependencies for~$B$,
i.e.\ equation~(\ref{eq:depB}).
Now, $B$ is not given by a single ROABP, but is a sum of~$c-1$ ROABPs.
For every $k \in [n]$ and $\b \in \depending_k(A)$,
define the polynomial 
$Q = \coeffset{B}{\y_k}{\b} - \sum_{\a \in \spanning_k(A)} \gamma_\a \coeffset{B}{\y_k}{\a}$.
By the definition of~$B$ we have
\begin{equation}
Q = 
\sum_{i=2}^c \left( \coeffset{A_i}{\y_k}{\b} - \sum_{\a \in \spanning_k(A)} \gamma_\a \coeffset{A_i}{\y_k}{\a} \right).
\label{eq:depcA}
\end{equation}
As explained in the previous section for equation~(\ref{eq:depB'}), 
for each summand in equation~(\ref{eq:depcA}) we can construct an ROABP of width~$w(w+1)$.
Thus, $Q$ can be written as a sum of~$c-1$ ROABPs, each having width~$w(w+1)$.
To test whether $Q = 0$, we recursively use the same algorithm for the sum of~$c-1$ ROABPs.
The recursion ends when $c=2$.
Then we directly use the algorithm from Section~\ref{sec:twoROABP}.

To bound the running time of the algorithm,
let us see how many dependencies we need to verify. 
There is one dependency for every $k \in [n]$ and every  $\b \in \depending_k(A)$.
Since $\abs{\depending_k(A)} \leq w(d+1)$, 
the total number of dependencies verified is $ \leq nw(d+1)$. 
Thus, we get the following recursive formula for~$T(c,w)$, 
the time complexity  for testing zeroness of the sum of~$c \geq 2$ ROABPs, each having width~$w$.
For $c=2$, we have $T(2,w) = \poly(n,d,w)$, and for $c >2$,
\[
T(c,w) = nw(d+1) \cdot T(c-1,w(w+1)) + \poly(n,d,w). 
\]
As solution, we get
$T(c,w) = w^{O(2^c)} \poly(n^c, d^c)$,
i.e.\ polynomial time for constant~$c$.

\begin{theorem}\label{thm:PITcROABPwb}
Let $A(\x)$ be an $n$-variate polynomial of individual degree~$d$, computed by a sum of~$c$ {\rm ROABPs} of width~$w$.
Then there is a PIT for~$A(\x)$ that works in time $w^{O(2^c)} (nd)^{O(c)}$.
\end{theorem}


\section{Blackbox Identity Testing}
\label{sec:blackbox}

In this section,
we extend the blackbox PIT of Agrawal et.~al~\cite{AGKS14} for one ROABP
to the sum of constantly many ROABPs.
In the blackbox model we are only allowed to evaluate a polynomial
at various points.
Hence, for PIT, our task is to construct a hitting-set.

\begin{definition}\label{def:hittingset}
A set $H = H(n,d,w) \subseteq \F^n$ is a {\em hitting-set for ROABPs\/},
if for every nonzero $n$-variate polynomial~$A(\x)$ of individual degree~$d$
that can be computed by ROABPs of width~$w$,
there is a point $\a \in H$ such that $A(\a) \not= 0$.

For polynomials computed by a sum of~$c$  ROABPs,
a hitting-set is defined similarly.
Here,  $H = H(n,d,w,c)$ additionally depends on~$c$.
\end{definition}

For a hitting-set to exist,
we will need enough points in the underlying field~$\F$.
Henceforth, we will assume that the field~$\F$ is large enough
such that the constructions below go through 
(see \cite{AJ86} for constructing large $\F$).
To construct a hitting-set for a sum of ROABPs 
we use the concept of {\em low support rank concentration\/} defined 
by Agrawal, Saha, and Saxena~\cite{ASS13}.
A polynomial~$A(\x)$  has  low support concentration
if  the coefficients of its monomials of low support span  the coefficients
of all the monomials. 

\begin{definition}[\cite{ASS13}]
\label{def:concentration}
A polynomial $A(\x)$ has $\ell$-{\em support concentration\/}
if for all monomials~$\xa$ of~$A(\x)$, we have,
\[
\coeff_A(\xa) \in \Span_\F \{\coeff_A(\xb) \mid  \supp(\b) < \ell \}.
\]
\end{definition}

The above definition applies to polynomials over any $\F$-vector space, e.g.\ 
$\F[\x]$, $\F^w[\x]$ or $\F^{w \times w}[\x]$.
Thus, $A(\x) \in \F[\x]$ is a non-zero polynomial that has $\ell$-support concentration
if and only if there are nonzero coefficients of support~$< \ell$.
An $\ell$-concentrated polynomial in $\F(\x)$ has the following hitting set.

\begin{lemma}[\cite{ASS13}]
\label{lem:hsFromlConc}
For $n,d,\ell$,
 the set~$H = \{ \h \in \{0,\beta_1,\dots,\beta_d\}^n \mid \supp(\h)< \ell\}$ 
of size $(n d)^{O(\ell)}$
is a hitting-set  for all $n$-variate $\ell$-concentrated polynomials 
$A(\x) \in \F[\x]$ of individual degree~$d$,
where $\{\beta_i\}_i$ are distinct nonzero elements in $\F$.
\end{lemma}

Hence,
when we have low support concentration, this solves blackbox PIT.
Note that every polynomial does not have low support concentration,
for example $A(\x) = x_1x_2 \cdots x_n$ is not $n$-concentrated.
However, 
Agrawal, Saha, and Saxena~\cite{ASS13} showed that low support concentration can be achieved
through an appropriate shift of the variables. 

\begin{definition}
Let $A(\x)$ be an $n$-variate polynomial and $\f = (f_1,f_2, \dots, f_n) \in \F^n$.
The polynomial~$A$ {\em shifted by\/}~$\f$ is 
$A(\x + \f) = A(x_1+f_1, x_2+f_2, \dots, x_n+f_n)$.
\end{definition}
Note that a shift is an invertible process.
Therefore it preserves the coefficient space of a polynomial.

In the above example, 
we shift every variable by~$1$.
That is, we consider $A(\x + {\bf 1}) = (x_1+1)(x_2+1)\cdots(x_n+1)$.
Observe that~$A(\x + {\bf 1})$ has $1$-support concentration.
A polynomial $A(\x)$ can also be shifted by polynomials.
Then, $\f$ would be a tuple of $n$ polynomials.
Agrawal, Saha, and Saxena~\cite{ASS13} provide an efficient shift that
achieves low support concentration for polynomials computed by {\em set-multilinear depth-3 circuits\/}.
Here, a shift is efficient if $\f$ itself
can be computed in quasi-polynomial time and
$A(\x + \f)$ has a hitting set that is computable in quasi-polynomial time.
Forbes, Saptharishi and Shpilka~\cite{FSS14} extended their result to polynomials computed 
by ROABPs. However their cost is exponential in the individual degree of the polynomial. 

Any efficient shift that achieves low support concentration
for ROABPs will suffice for our purposes. 
In Section~\ref{sec:conc}, we will give a new shift for ROABPs with quasi-polynomial cost.  
Namely,
in Theorem~\ref{thm:kronecker} below we present a shift polynomial~$\f(t) \in \F[t]^n$ 
in one variable~$t$ of degree~$(ndw)^{O(\log n)}$ 
that can be computed in time~$(ndw)^{O(\log n)}$.
It has the property that for every $n$-variate polynomial~$A(\x) \in \F^{w \times w}[\x]$ of individual degree~$d$
that can be computed by  an ROABP of width~$w$,
the shifted polynomial $A(\x + \f(t))$ has 
$O(\log w)$-concentration. 
We can plug in as many values for~$t \in \F$ as the degree of~$\f(t)$,
i.e.\ $(ndw)^{O(\log n)}$ many.
For at least one value of~$t$, the shift~$\f(t)$ will $O(\log w)$-concentrate
$A(\x + \f(t))$.
That is,
we consider~$\f(t)$ as a family of shifts.
The same shift also works when the ROABP computes a polynomial
in $\F[\x]$ or $\F^{1 \times w}[\x]$.

The rest of the paper is organized as follows.
The construction of a shift to obtain low support concentration for single ROABPs
is postponed to Section~\ref{sec:conc}.
We start in Section~\ref{sec:sumOfC-bb} to show how the shift for a single ROABP
can be applied to obtain a shift for the sum of constantly many ROABPs.


\subsection{Sum of ROABPs}
\label{sec:sumOfC-bb}

We will first give a hitting set for the sum of two ROABPs, $A + B$.
We will then extendthis result for the sum of $c$ ROABPs.
Let polynomial $A \in \F[\x]$ of individual degree~$d$ have  an ROABP of width~$w$,
with variable order $(x_1, x_2, \dots, x_n)$.
Let $B\in \F[\x]$ be another polynomial.
We start by reconsidering the whitebox test from the previous section.
The dependency equations~(\ref{eq:depA}) and~(\ref{eq:depB})
were used to construct an ROABP for $B\in \F[\x]$ in the same variable order as for~$A$,
and the same width.
If this succeeds,
then the polynomial $A + B$ has one ROABP of width~$2w$.
Since there is already a blackbox PIT for one ROABP~\cite{AGKS14},
we are done in this case.

Hence, the interesting case that remains is
when the dependency equations~(\ref{eq:depA}) for~$A$ do not carry over to~$B$ as in equation~(\ref{eq:depB}).
Let $k \in [n]$ be the first such index.
In the following Lemma~\ref{lem:indepRi} we decompose~$A$ and~$B$
into a common ROABP $R$ up to layer~$k$,
and the remaining different parts $P$ and $Q$.
That is, for $\y_k = (x_1,x_2, \dots, x_k)$ and $\z_k = (x_{k+1}, \dots, x_n)$,
we obtain $A = RP$ and $B = RQ$,
where $R \in \F[\y_k]^{1 \times w'}$  
and $P,Q \in \F[\z_k]^{w' \times 1}$,
for some $w' \leq w(d+1)$.
The construction we give is such that that the coefficient space of~$R$ has full rank~$w'$.
Since the dependency equations~(\ref{eq:depA}) for~$A$ do not fulfill equation~(\ref{eq:depB}) for~$B$,
we get a constant vector $\Gamma \in \F^{1 \times w'}$ such that
$\Gamma P =0$ but $\Gamma Q \neq 0$. 
From these properties,
we will see in Lemma~\ref{lem:blackBoxSumOfTwo} below
that we get low support concentration for~$A+B$
when we use the shift constructed in Section~\ref{sec:conc} for one ROABP.

\begin{lemma}[Common ROABP $R$]
\label{lem:indepRi}
Let $A(\x)$ be a polynomial of individual degree~$d$,
computed by an ROABP of width~$w$ in variable order  $(x_1,x_2, \dots, x_n)$.
Let~$B(\x)$ be another polynomial for which 
there does {\em not\/} exist an ROABP of width~$w$ in the same variable order.

Then there exists a $k \in [n]$ such that for some $w' \leq w(d+1)$,
there are polynomials
$R \in \F[\y_k]^{1 \times w'}$ and $P,Q \in \F[\z_k]^{w' \times 1}$, 
 such that
\begin{enumerate}
\item 
$A = RP$ and $B = RQ$,
\item 
there exists a vector $\Gamma \in \F^{1 \times w'}$ with $\supp(\Gamma) \leq w+1$
such that $\Gamma P = 0$
and $\Gamma Q \neq 0$,
\item 
the coefficient space of~$R$ has full rank~$w'$.
\end{enumerate}
\end{lemma}

\begin{proof}
Let $D_1, D_2, \dots, D_n$ be the matrices constructed in Lemma~\ref{lem:dimROABP} for~$A$. 
Assume again w.l.o.g.\ that 
$\spanning_k(A) = \{ \a_{k,1},\a_{k,2}, \dots, \a_{k,w} \}$ has size~$w$
for each $1 \leq k \leq n-1$,
and $\spanning_n(A) = \{ \a_{n,1} \}$.
Then we have $D_1 \in \F^{1 \times w}[x_1]$, $D_n \in \F^{w \times 1}[x_n]$
and $D_i \in \F^{w \times w}[x_i]$, for $2 \leq i \leq n-1$.

In the proof of Lemma~\ref{lem:dimROABP} we consider the dependency equations for~$A$
and carry them over to~$B$.
By the assumption of the lemma,
there is no ROABP of width~$w$ for~$B$ now.
Therefore there is a smallest $k \in [n]$
where a dependency for~$A$ is not followed by~$B$.
That is,
the coefficients~$\gamma_{\a}$ computed for equation~(\ref{eq:depA}) 
do not fulfill equation~(\ref{eq:depB}) for~$B$.
Since the dependencies carry over up to this point,
the construction of the matrices $D_1, D_2, \dots, D_{k-1}$ work out fine for~$B$.
Hence,
by equation~(\ref{eq:D1-k}), 
we can write
\begin{eqnarray}
 A(\x) &=& D_1\, D_2 \cdots D_{k-1} \,
[\coeffset{A}{\y_{k-1}}{\a_{{k-1},1}} \, \coeffset{A}{\y_{k-1}}{\a_{{k-1},2}} \, \cdots \, \coeffset{A}{\y_{k-1}}{\a_{{k-1},w}} ]^T 
\label{eq:Ak} \\
 B(\x) &=& D_1\, D_2 \cdots D_{k-1}  \,
[\coeffset{B}{\y_{k-1}}{\a_{{k-1},1}} \, \coeffset{B}{\y_{k-1}}{\a_{{k-1},2}} \, \cdots \, \coeffset{B}{\y_{k-1}}{\a_{{k-1},w}} ]^T 
\label{eq:Bk}
\end{eqnarray}

Since the difference between~$A$ and~$B$ occurs at~$x_k$,
we consider all possible extensions from~$x_{k-1}$.
That is,
by equation~(\ref{eq:Apartk}),
for every $i \in  [w]$ we have
\begin{eqnarray}
\coeffset{A}{\y_{k-1}}{\a_{{k-1},i}} &=& \sum_{j=0}^d \coeffset{A}{\y_{k}}{ (\a_{{k-1},i}, j) } x_{k}^j \,.
\label{eq:partk1}
\end{eqnarray}

Recall that our goal is to decompose polynomial~$A$ into $A = R P$.
We first define  polynomial~$P$ as 
the vector of coefficient polynomials of all the one-step extensions of $\spanning_{k-1}(A)$, 
i.e.,
$P = \left(\coeffset{A}{\y_{k}}{ (\a_{{k-1},i}, j) }\right)_{1 \leq i \leq w,~ 0 \leq j \leq d}$
is  of length $w' = w(d+1)$.
Written explicitly,
this is
\[
P = [\coeffset{A}{\y_{k}}{ (\a_{{k-1},1}, 0) } \cdots \coeffset{A}{\y_{k}}{ (\a_{{k-1},1}, d)} ~\cdots~ 
\coeffset{A}{\y_{k}}{ (\a_{{k-1},w}, 0) }  \cdots \coeffset{A}{\y_{k}}{ (\a_{{k-1},w}, d) }]^T \,.
\]
To define~$R \in \F[\y_k]^{1 \times w'} $,
let~$I_w$ be the $w \times w$ identity matrix.
Define matrix $E_k \in \F[x_k]^{w \times w'}$
as the tensor product
\[
E_k = I_{w} \otimes \sqbrace{x_k^0 \; x_k^1 \, \cdots \, x_k^d}\,.
 \]
From equation~(\ref{eq:partk1}) we get that
\[
[\coeffset{A}{\y_{k-1}}{\a_{{k-1},1}} \cdots \coeffset{A}{\y_{k-1}}{\a_{{k-1},w}}]^T
=
E_k\, P.
\]

Thus, equation~(\ref{eq:Ak}) can be written as
$ A(\x) = D_1\, D_2 \cdots D_{k-1} E_k P$.
Hence, 
when we define 
\[
R(\y_k)  = D_1\, D_2 \cdots D_{k-1} E_k  
\]
then we have $A = R P$ as desired.
By an analogous argument we get
$B = R Q$ for 
$Q = \left(\coeffset{B}{\y_{k}}{ (\a_{{k-1},i}, j) }\right)_{1 \leq i \leq w,~ 0 \leq j \leq d}$.

For the second claim of the lemma
let $\b \in \depending_k(A)$ such that the dependency equation~(\ref{eq:depA}) for~$A$ is fulfilled,
but not equation~(\ref{eq:depB}) for~$B$.
Define $\Gamma \in \F^{1 \times w'}$ to be the vector 
that has the values~$\gamma_{\a}$ used in equation~(\ref{eq:depA})
at the position where~$P$ has entry~$A_{(\y_k,\a)}$,
and zero at all other positions.
Then $\supp(\Gamma) \leq w+1$ and we have
$\Gamma P = 0$ and $\Gamma Q \neq 0$.

It remains to show that the coefficient space of~$R$ has full rank. 
By Corollary~\ref{cor:coeffspaceP},
the coefficient space of~$D_1\, D_2 \cdots D_{k-1}$ has full rank~$w$.
Namely, for any $\ell \in [w]$,
the coefficient of the monomial~$\y_{k-1}^{\a_{k-1,\ell}}$ is~$\e_{\ell}$,
the $\ell$-th standard unit vector.
Therefore the coefficient of $R(\y_k)  = D_1\, D_2 \cdots D_{k-1} E_k$
at monomial $\y_k^{(\a_{k-1,\ell},j)}$ is
\[
\coeff_R(\y_{k}^{\a_{{k-1},\ell},j} ) = \e_{\ell}\, \coeff_{E_k}(x_k^j),
\]
for $1\leq \ell \leq w$ and $0 \leq j \leq d$.
By the definition of~$E_k$, we get  
$\coeff_R(\y_{k}^{\a_{{k-1},\ell},j} )= e_{(\ell-1)(d+1) +j +1}$.
Thus, the coefficient space of~$R$ has full rank~$w'$.
\end{proof}

Lemma~\ref{lem:indepRi} provides the technical tool to obtain
low support concentration for the sum of several ROABPs
by the shift developed for a single ROABP.
We start with the case of the sum of two ROABPs.

\begin{lemma}
\label{lem:blackBoxSumOfTwo}
Let $A(\x)$ and $B(\x)$ be two $n$-variate polynomials of individual degree~$d$,
each computed by an  $\ROABP$ of width~$w$.
Define $W_{w,2} = (d+1)(2w)^2$ and $\ell_{w,2} = \log( W_{w,2}^2 + 1)$.
Let $\f_{w,2}(t) \in \F[t]^n$ be a shift that $\ell_{w,2}$-concentrates
any polynomial (or matrix polynomial)
that is computed by an $\ROABP$ of width~$\leq W_{w,2}$.

Then $(A + B)' = (A + B)(\x + \f_{w,2})$ is $2\,\ell_{w,2}$-concentrated.
\end{lemma}

\begin{proof}
If $B$ can be computed by an ROABP of width~$w$ in the same variable order as the one for~$A$,
then there is an ROABP of width~$2w$ that computes $A+B$.
In this case, the lemma follows because $2w \leq W_{w,2}$.
So let us assume that there is no such ROABP for~$B$.
Thus the assumption from Lemma~\ref{lem:indepRi} is fulfilled.
Hence,
we have a decomposition of~$A$ and~$B$ at the $k$-th layer into $A(\x) = R(\y_k)P(\z_k)$ and $B(\x) = R(\y_k)Q(\z_k)$,
and there is a vector~$\Gamma \in \F^{1 \times w'}$ 
such that $\Gamma P = 0$ and $\Gamma Q \not= 0$,
where $w' = (d+1)w$ and $\supp(\Gamma) \leq w+1$.

Define
$R',P',Q'$ as the polynomials $R,P,Q$ shifted by~$\f_{w,2}$, respectively.
Since $\Gamma P = 0$, we also have $\Gamma P' = 0$.

By the definition of~$R$,
there is an ROABP of width~$w'$ that computes~$R$.
Since $w' \leq W_{w,2}$, 
polynomial~$R'$ is $\ell_{w,2}$-concentrated by the assumption of the lemma.

We argue that also $\Gamma Q'$ is $\ell_{w,2}$-concentrated:
let $Q = [Q_1 \, Q_2  \cdots  Q_{w'}]^T\in \F[\z_k]^{w' \times 1}$.
By Lemma~\ref{lem:pdROABP},
from the ROABP for~$B$ we get an ROABP for each~$Q_i$ of the same width~$w$
and the same variable order.
Therefore we can combine them into one ROABP that computes
$\Gamma Q = \sum_{i=1}^{w'} \gamma_i Q_i$.
Its width is~$w(w+1) $ because $\supp(\Gamma) \leq w+1$.
Since $w(w+1) \leq W_{w,2}$, 
polynomial $\Gamma Q'$ is $\ell_{w,2}$-concentrated.

Since $\Gamma Q \ne 0$
and $\Gamma Q'$ is $\ell_{w,2}$-concentrated,
there exists at least one $\b \in \{0,1, \dots, d\}^{n-k}$ with $\supp(\b) < \ell_{w,2}$ 
such that $\Gamma \coeff_{Q'} (\zkb) \ne 0$.
Because $\Gamma P = 0$,
we have 
$\Gamma \coeff_{P'}(\zkb) = 0$,
and therefore
\begin{equation}
\label{eq:nonZeroMonomial}
 \Gamma \coeff_{P' + Q'}(\zkb) \ne 0.
\end{equation}

Recall that the coefficient space of~$R$ has full rank~$w'$.
Since a shift preserves the coefficient space,
$R'$ also has a full rank coefficient space.
Because~$R'$ is $\ell_{w,2}$-concentrated,
already the coefficients of the $<\ell_{w,2}$-support monomials of~$R'$ have full rank~$w'$.
That is, for 
$M_{\ell_{w,2}} = \{\a \in \{0,1, \dots,d\}^k \mid \supp(\a) < \ell_{w,2}\}$, 
we have
$\rank_{\F(t)}\{\coeff_{R'}(\yka) \mid \a \in M_{\ell_{w,2}} \} = w'$.
Therefore, we can express~$\Gamma$ as a linear combination of these coefficients,
\[
 \Gamma = \sum_{\a \in M_{\ell_{w,2}}} \alpha_\a  \coeff_{R'} (\yka),
\]
where $\alpha_\a$ is a rational function in~$\F(t)$, for $\a \in M_{\ell_{w,2}}$.
Hence, from equation~(\ref{eq:nonZeroMonomial}) we get
\begin{eqnarray*}
\Gamma \coeff_{(P' + Q')}(\zkb) 
&=& \left(\sum_{\a \in M_{\ell_{w,2}}} \alpha_\a  \coeff_{R'} (\yka) \right)  \coeff_{P' + Q'}(\zkb) \\
&=& \sum_{\a \in M_{\ell_{w,2}}} \alpha_\a  \coeff_{R' (P' + Q')} (\yka\, \zkb) \\
&=& \sum_{\a \in M_{\ell_{w,2}}} \alpha_\a \coeff_{(A+B)'} (\x ^{(\a, \b)})\\
&\not=& 0\,.
\end{eqnarray*}
Since $\supp(\a , \b) = \supp(\a) + \supp(\b) <  2 \ell_{w,2}$,
it follows that there is a monomial in~$(A+B)'$ of support $< 2\ell_{w,2}$  with a nonzero coefficient.
In other words, $(A+B)'$ is $2\ell_{w,2}$-concentrated.
\end{proof}

In Section~\ref{sec:conc}, Theorem~\ref{thm:kronecker},
we will show that the shift polynomial~$\f_{w,2}(t) \in \F[t]^n$
used in Lemma~\ref{lem:blackBoxSumOfTwo}
can be computed in time $ (ndw)^{O(\log n)}$.
The degree of~$\f_{w,2}(t)$ is also $(ndw)^{O(\log n)}$.
Recall that when we say that we shift by~$\f_{w,2}(t)$,
we actually mean that we plug in values for~$t$ up to the degree of~$\f_{w,2}(t)$.
That is,
we have a family of $ (ndw)^{O(\log n)}$ shifts,
and at least one of them will give low support concentration.
By Lemma~\ref{lem:hsFromlConc}, we get for each~$t$, a potential
hitting-set $H_t$ of size $(n d)^{O(\ell_{w,2})} = (n d)^{O(\log d w)}$,
\[
H_t = \{ \h + \f(t)  \mid \h \in \{0,\beta_1,\dots,\beta_d\}^n \text{ and } \supp(\h) < 2\ell_{w,2}\}\,.
\]
The final hitting-set is the union of all these sets, 
i.e.\ $H = \bigcup_{t} H_t$,
where $t$ takes $(ndw)^{O(\log n)}$ distinct values.
Hence, we have the following main result.

\begin{theorem}
\label{thm:blackbox2}
Given $n,d,w$,
in time $(ndw)^{O(\log ndw)}$ one can construct a hitting-set for
all  $n$-variate polynomials of individual degree~$d$,
that can be computed by a sum of two  ROABPs of width~$w$. 
\end{theorem}

We extend Lemma~\ref{lem:blackBoxSumOfTwo} to the sum of~$c$ ROABPs.

\begin{lemma}
\label{lem:blackBoxSumOfC}
Let $A = A_1 + A_2 + \cdots + A_c$,
where the $A_i$'s are $n$-variate polynomials of individual degree~$d$,
each computed by an $\ROABP$ of width~$w$.
Define $W_{w,c} = (d+1)\wt{c-1}$ and $\ell_{w,c} = \log( W_{w,c}^2+1)$.
Let $\f_{w,c}(t) \in \F[t]^n$ be a shift that $\ell_{w,c}$-concentrates
any polynomial (or matrix polynomial)
that is computed by an $\ROABP$ of width~$W_{w,c}$.

Then 
$A' = A(\x + \f_{w,c})$ is $c\, \ell_{w,c}$-concentrated.
\end{lemma}

\begin{proof}
The proof is by induction on~$c$.
Lemma~\ref{lem:blackBoxSumOfTwo} provides the base case $c=2$.
For the induction step let $c \geq 3$.
We follow the proof of Lemma~\ref{lem:blackBoxSumOfTwo}
with $A = A_1$ and $B = \sum_{j = 2}^c A_j$.
Consider again the decomposition of~$A$ and~$B$ at the $k$-th layer into 
$A = RP$ and $B = RQ$,
and let~$\Gamma \in \F^{1 \times w'}$ such that $\Gamma P = 0$ and $\Gamma Q \not= 0$,
where $w' = (d+1)w$ and $\supp(\Gamma) \leq w+1$.

The only difference to the proof of Lemma~\ref{lem:blackBoxSumOfTwo} is~$Q = [Q_1 \, Q_2  \cdots  Q_{w'}]^T$.
Recall from Lemma~\ref{lem:indepRi} that 
$Q_i = \coeffset{B}{\y_k}{\a_i} = \sum_{j = 2}^c \coeffset{A_j}{\y_k}{\a_i}$, 
for $\a_i \in \depending_k(A)$.
Hence, 
\[
\Gamma Q 
=\sum_{i = 1}^{w'} \gamma_i \left(\sum_{j = 2}^c \coeffset{A_j}{\y_k}{\a_i} \right)
=\sum_{j = 2}^c \sum_{i = 1}^{w'} \gamma_i \coeffset{A_j}{\y_k}{\a_i} \,.
\]
By Lemma~\ref{lem:pdROABP},  $\Gamma Q$ can be computed by a sum of~$c-1$ ROABPs,
each of width~$w (w+1) \leq 2 w^2 = w''$, because $\supp(\Gamma) \leq w+1$.
Our definition of~$W_{w,c}$ was chosen such that 
\[
W_{w'',c-1} 
= (d+1)(2w'')^{2^{c-2}} 
=  (d+1)(2 \cdot 2w^2)^{2^{c-2}} 
= (d+1)(2w)^{2^{c-1}} 
=W_{w,c} \,.
\]
Hence, $\f_{w,c}(t)$ is a shift that $\ell_{w'',c-1}$-concentrates
any polynomial
that is computed by an $\ROABP$ of width~$W_{w'',c-1}$.
By the induction hypothesis,
we get that
$\Gamma Q' = \Gamma Q (\x +\f_{w,c}(t))$ is $(c-1)\,\ell_{w'',c-1}$-concentrated,
which is same as $(c-1)\,\ell_{w,c}$-concentrated.

Now we can proceed as in the proof of Lemma~\ref{lem:blackBoxSumOfTwo}
and get that
$(A + B)' = \sum_{j=1}^c A'_j$ has a monomial of support 
$< \ell_{w,c} + (c-1)\, \ell_{w,c} = c\, \ell_{w,c}$.
\end{proof}

We combine the lemmas similarly as for Theorem~\ref{thm:blackbox2}
and obtain our main result for the sum of constantly many ROABPs.

\begin{theorem}
\label{thm:blackbox}
Given $n,w, d$,
in time $(ndw)^{O(c \cdot 2^c \log ndw)}$ one can construct a hitting-set for
all  $n$-variate polynomials of individual degree~$d$,
that can be computed by the sum of~$c$  ROABPs of width~$w$. 
\end{theorem}


\subsection{Concentration in matrix polynomials}

As a by-product, we show that low support concentration can be 
achieved even when we have a sum of matrix polynomials, 
each computed by an ROABP.
For a matrix polynomial $A(\x) \in F^{w \times w}[\x]$,
an ROABP is defined similar to the standard case.
We have layers of nodes $V_0, V_1, \dots, V_n$ connected by directed edges 
from~$V_{i-1}$ to~$V_i$.
Here, $V_0 = \{ v_{0,1},v_{0,2}, \dots,v_{0,w} \}$ and
$V_n = \{ v_{n,1},v_{n,2}, \dots,v_{n,w} \}$ also consist of~$w$ nodes.
The polynomial~$A_{i,j}(\x)$ at position~$(i,j)$ in~$A(\x)$
is the polynomial computed by the standard ROABP with start node~$v_{0,i}$
and end node~$v_{n,j}$.

Note that Definition~\ref{def:concentration} for $\ell$-support concentration
can be applied to polynomials over any $\F$-algebra.

\begin{corollary}
\label{cor:concOverAlgebra}
Let $A = A_1 + A_2 + \cdots + A_c$,
where each $A_i \in \F^{w \times w}[\x]$ is an $n$-variate matrix polynomials of individual degree~$d$,
each computed by an $\ROABP$ of width~$w$.
Let $\f_{w,c}$ and $\ell_{w, c}$ be defined as in Lemma~\ref{lem:blackBoxSumOfC}.

Then $A(\x + \f_{w^2, c})$ is $c\ell_{w^2, c}$-concentrated.
\end{corollary}

\begin{proof}
Let $\alpha \in \F^{w \times w}$ and consider the dot-product 
$\abrace{\alpha, A_i} \in \F[\x]$.
This polynomial can be computed by an ROABP of width~$w^2$:
we take the ROABP of width~$w$ for~$A_i$ and make~$w$ copies of it,
and two new nodes~$s$ and~$t$.
We add the following edges.
\begin{itemize}
\item
Connect the new start node~$s$ to the $h$-th former start node
of the $h$-th copy of the ROABP
by edges of weight one,
for all $1 \le h \le w$.
\item
Connect the $j$-th former end node of the $h$-th copy of the ROABP
to the new end node~$t$ by an edge of weight~$\alpha_{h,j}$,
for all $1 \le h, j \le w$.
\end{itemize}
The resulting ROABP has width~$w^2$ and computes~$\abrace{\alpha, A_i}$.

Now consider the polynomial
$\abrace{\alpha, A} = \abrace{\alpha, A_1} + \abrace{\alpha, A_2} +
\cdots + \abrace{\alpha, A_c} $.
It can be computed by a sum of~$c$ ROABPs, each of width~$w^2$,
for every $\alpha \in \F^{w \times w}$.
Hence, by Lemma \ref{lem:blackBoxSumOfC},
the polynomial $\abrace{\alpha, A} (\x + \f_{w^2, c})$
is $c\ell_{w^2, c}$-concentrated, for every $\alpha \in \F^{w \times w}$.
By Lemma~\ref{lem:vectorConc} below,
it follows that
$A(\x + \f_{w^2, c})$ is $c\ell_{w^2, c}$-concentrated.
\end{proof}

The following lemma is also of independent interest.

\begin{lemma}
\label{lem:vectorConc}
Let~$A \in \F^{w \times w}[\x]$ be an $n$-variate polynomial 
and~$\f(t)$ be a shift.
Then
$A(\x + \f(t))$ is $\ell$-concentrated iff ~$\forall \alpha \in \F^{w \times w},$
 $\abrace{\alpha, A} (\x + \f(t))$ is $\ell$-concentrated.
\end{lemma}

\begin{proof}
Assume that $A'(\x) = A(\x + \f)$ is not $\ell$-concentrated.
Then there exists a monomial~$\x^{\b}$ such that
$ \coeff_{A'}(\xb) \notin
\Span_{\F(t)}\{\coeff_{A'}(\xa) \mid \supp(\a) < \ell\}$.
Hence, 
there exists an $\alpha  \in \F^{w \times w}$ such that 
$\abrace{\alpha, \coeff_{A'}(\xa)} = 0$, for all~$\a$ with $\supp(\a) < \ell$,
but $\abrace{\alpha, A'} \ne 0 $.
We thus found an $\alpha  \in \F^{w \times w}$  such that
$\abrace{\alpha, A} (\x + \f(t))$ is not $\ell$-concentrated.

For the other direction, let $A(\x + \f)$ be $\ell$-concentrated.
Hence, any coefficient~$\coeff_{A'}(\xa)$ can be written as a linear combination
of the small support coefficients,
\[
\coeff_{A'}(\xa) = \sum_{\substack{\b \\ \supp(\b) < \ell}} \gamma_{\b}\, \coeff_{A'}(\xb),
\]
for some $\gamma_{\b} \in \F$.
Hence,
for any $\alpha \in \F^{w \times w}$, we also have
\[
\abrace{\alpha,\coeff_{A'}(\xa)} 
= \abrace{\alpha,\sum_{\substack{\b \\ \supp(\b) < \ell}} \gamma_{\b}\, \coeff_{A'}(\xb)}.
\]
That is,
$\abrace{\alpha, A} (\x + \f(t))$ is $\ell$-concentrated.
\end{proof}


\section{Low Support Concentration in ROABPs}
\label{sec:conc}

Recall that a polynomial~$A(\x)$ over an $\F$-algebra~$\A$ is called low-support
concentrated if its low-support coefficients span all its coefficients. 
We show an efficient shift which achieves concentration in matrix polynomials computed by ROABPs. 
We use the quasi-polynomial size hitting-set for ROABPs
given by Agrawal~et~al.~\cite{AGKS14}.
Their hitting-set is based on a {\em basis isolating weight assignment\/}
which we define next.

Recall that $\M = \{0,1, \dots,d\}^n$ denotes the set of all exponents of monomials in~$\x$ 
of individual degree bounded by~$d$.
For a weight function $\w \colon [n] \to \N$ and $\a = (a_1,a_2, \dots, a_n) \in \M$, let
the weight of~$\a$ be $\w(\a) = \sum_{i=1}^n \w(i)a_i$.
Let $\A_k$ be a $k$-dimensional algebra over the field $\F$.

\begin{definition}
A weight function $\w \colon [n] \to \N$ is called a 
{\em basis isolating weight assignment for a polynomial\/} $A(\x) \in \A_k[\x]$,
if there exists $S \subseteq \M$ with $\abs{S} \leq k$ such that
\begin{itemize}
\item  
$\forall \, \a \not= \b \in S,~~\w(\a) \neq \w(\b)$ and
\item 
$ \forall\, \a \in \overline{S} := \M - S, ~~\coeff_A(\xa) \in \Span_\F \{ \coeff_A(\xb) \mid \b \in S \text{ and } \w(\b) < \w(\a) \}.$
\end{itemize}
\end{definition}

Agrawal~et~al.~\cite[Lemma 8]{AGKS14} presented a quasi-polynomial time construction of 
such a weight function for any polynomial~$A(\x) \in \F^{w \times w}[\x]$ computed by an ROABP.
The hitting-set is then defined by points
$(t^{\w(1)}, t^{\w(2)}, \dots, t^{\w(n)})$ for $\poly(n,d,w)^{\log n}$ many~$t$'s.
Our approach now is to use this weight function for a shift of~$A(\x)$ by $\left(t^{\w(i)}\right)_{i=1}^n$.
Let~$A'(\x)$ denote the shifted polynomial,
\[
 A'(\x) = A(\x + t^\w) = A(x_1 +t^{\w(1)}, x_2 + t^{\w(2)}, \dots, x_n+t^{\w(n)})\,.
\]
We will prove that~$A'$ has low support concentration.

The coefficients of~$A'$ are linear combinations
of coefficients of~$A$, which are given by the equation
\begin{equation}
\label{eq:shift}
\coeff_{A'}(\xa) = \sum_{\b \in \M} \binom{\b}{\a} t^{\w(\b-\a)} \cdot \coeff_{A}(\xb),
\end{equation}
where $\binom{\b}{\a} = \prod_{i=1}^n \binom{b_i}{a_i} $ for any $\a,\b \in \N^n$.

Equation~(\ref{eq:shift}) can be expressed in terms of matrices.
Let~$C$ be the coefficient matrix of~$A$,
i.e.\ the $\M \times [k]$ matrix with the coefficients~$\coeff_{A}(\xa)$ as rows,
\[
 \row{C}{\a} = \coeff_{A}(\xa)^T \,.
\]
Similarly, 
let~$C'$ be the $\M \times [k]$ with the coefficients~$\coeff_{A'}(\xa)$ as rows.
Let furthermore~$T$ be the $\M \times \M$ {\em transfer matrix\/} given by
\[
\entry{\T}{\a}{\b} = \binom{\b}{\a} \,,
\]
and~$D$ be the $\M \times \M$ diagonal matrix given by 
\[
\entry{D}{\a}{\a} = t^{\w(\a)} \,.
\]
The inverse of~$D$ is the diagonal matrix given by $\entry{D^{-1}}{\a}{\a} = t^{-\w(\a)}$.
Now equation~(\ref{eq:shift}) becomes
\begin{equation}
\label{eq:shiftMatrix}
C' = D^{-1} T D C \,.
\end{equation}
As shifting is an invertible operation, the matrix~$\T$ is also invertible
and $\rank(C') = \rank(C)$.

\begin{lemma}[Isolation to concentration]
Let $A(\x)$ be a polynomial over a $k$-dimensional algebra~$\A_k$.
Let~$\w$ be a basis isolating weight assignment for~$A(\x)$.
Then $A(\x + t^\w)$ is $\ell$-concentrated, where $\ell = \ceil{\log(k+1)}$.
\label{lem:lconc}
\end{lemma}

\begin{proof}
Let $A'(\x) = A(\x + t^\w)$. 
We reconsider equation~(\ref{eq:shiftMatrix}) with respect to the low support monomials:
let $\Ml = \{\a \in \M \mid \supp(\a) < \ell \}$ be the exponents of low support.
Then we define matrices

\begin{center}
\begin{tabular}{rp{0.8\textwidth}}
$C'_\ell$ :& the $\Ml \times [k]$ submatrix of~$C'$ that
contains the  coefficients of~$A'$ of support~$<\ell$,\\
$\T_{\ell}$ :& the $\Ml \times \M$ submatrix of~$\T$ 
restricted to the rows $ \a \in \Ml$,\\
$D_{\ell}$ :& the $\Ml \times \Ml$ submatrix of~$D$ restricted
to the rows and columns from~$\Ml$.
\end{tabular}
\end{center}
To show that~$A'$ is $\ell$-concentrated, we need to prove that 
$\rank(C'_{\ell}) = \rank(C)$.
By equation~(\ref{eq:shiftMatrix}), matrix~$C'_\ell$ can be written as 
$C'_\ell = D_{\ell}^{-1} \T_{\ell} D C.$
Since~$D_{\ell}$ and~$D_{\ell}^{-1}$ are diagonal matrices, they have full rank.
Hence, it suffices to show that  $\rank(\T_\ell D C) = \rank(C)$.

W.l.o.g.\ we assume
that the order of the rows and columns in all the above matrices
that are indexed by~$M$ or~$\Ml$ is according to increasing weight~$\w(\a)$
of the indices~$\a$.
The rows with the same weight can be arranged in an arbitrary order.

Now, recall that~$\w$ is a basis isolating weight assignment.
Hence, there exists a set $S \subseteq \M$ 
such that the coefficients~$\coeff_A(\b)$, for $\b \in S$, span all 
coefficients~$\coeff_A(\a)$, for $\a \in \M$.
In terms of the coefficient matrix~$C$, for any $\a \in  \M$ we can write
\begin{equation} 
\row{C}{\a} \in \Span \{ \row{C}{\b} \mid \b \in S \text{ and } \w(\b) < \w(\a) \}.
\label{eq:lowWeightSpan}
\end{equation}

Let $S = \{\s_1, \s_2, \dots, \s_{k'} \}$ for some $k' \leq k$.
Let~$C_0$ be the $k' \times k$ submatrix of~$C$ 
whose $i$-{th} row is $\row{C}{\s_i}$,
i.e.\ $\row{C_0}{i} = \row{C}{\s_i}$.
By~(\ref{eq:lowWeightSpan}),
for every $\a \in \M$,
there is a vector $\bfgamma_{\a} = (\gamma_{\a,1}, \gamma_{\a,2}, \dots, \gamma_{\a,k'}) \in \F^{k'}$
such that
$\row{C}{\a} = \sum_{j=1}^{k'} \gamma_{\a,j}\, \row{C_0}{j}$.
Let~$\Gamma = \left(\gamma_{\a,j}\right)_{\a,j}$ be the $\M \times [k']$ matrix with these vectors as rows.
Then we get
\[
C = \Gamma C_0\,.
\]

Observe that the $\s_i$-th row of~$\Gamma$ is simply~$\e_i$, the $i$-th standard unit vector.
By~(\ref{eq:lowWeightSpan}),
the coefficient~$\row{C}{\s_i}$ is used to express~$\row{C}{\a}$ only when $\w(\a) > \w(\s_i)$.
Recall that the rows of the matrices indexed by~$M$, like~$\Gamma$,  
are in order the of increasing weight of the index.
Therefore,
when we consider the $i$-th column of~$\Gamma$ from the top,
the entries are all zero down to row~$s_i$, where we hit on the one from~$\e_i$,
\begin{equation}\label{eq:zero}
 \Gamma(\s_i,i) = 1  ~\text{ and }~ \forall \, \a \not= \s_i,~ \w(\a) \leq \w(\s_i) \implies \Gamma(\a,i) = 0 \,.
\end{equation}

Recall that our goal is to show $\rank(\Tl D C) = \rank(C)$.
For this,
it suffices to show that the $\Ml \times k'$ matrix $R = \Tl D\Gamma$ has full 
column rank~$k'$,
because then we have
$\rank(\Tl D C) = \rank(\Tl D \Gamma C_0) = \rank(R C_0) = \rank(C_0) = \rank(C)$.

To show that $R$ has full column rank~$k'$,
observe that the $j$-th column of~$R$ can be written as
\begin{equation}\label{eq:colR}
\column{R}{j} = \sum_{\a \in \M} \column{\Tl}{\a}\, \entry{\Gamma}{\a}{j}\, t^{\w(\a)} \,.
\end{equation}
By~(\ref{eq:zero}),
the term with the lowest degree in equation~(\ref{eq:colR}) is~$t^{\w(\s_j)}$.
By $\lc(\column{R}{j})$ we denote the coefficient of the lowest 
degree term in the polynomial $\column{R}{j}$.
Because $\entry{\Gamma}{\s_j}{j}=1$, we have
\[
 \lc(\column{R}{j}) = \column{\Tl}{\s_j} \,.
\]
We define the $\Ml \times [k']$ matrix~$R_0$ whose
$j$-th column is  $\lc(\column{R}{j})$,
i.e.\ $\column{R_0}{j}) = \column{\Tl}{\s_j}$.
We will show in Lemma~\ref{lem:transferMatrix} below that 
the columns of matrix~$\Tl$ indexed by the set~$S$ are linearly independent.
Therefore the~$k'$ columns of~$R_0$ are linearly independent. 

Hence, there are~$k'$ rows in~$R_0$ such that
its restriction to these rows, say~$R_0'$, is a square matrix with nonzero determinant. 
Let~$R'$ denote the restriction of~$R$ to the same set of rows. 
Now observe that the lowest degree term in~$\det(R')$ 
has coefficient precisely~$\det(R'_0)$,
i.e., $\lc(\det(R')) = \det(R'_0)$. 
This is because the lowest degree term in~$\det(R')$ has 
degree $\sum_{j=1}^{k'} \w(\s_j)$,
and this degree can only be obtained 
when the degree $\w(\s_j)$ term is taken from the $j$-th column, for all $j$.
We conclude that $\det(R') \neq 0$ and hence~$R$ has full column rank.
\end{proof}

It remains to show that  
the $k' \leq k$ columns of matrix~$\Tl$ 
indexed by the set~$S$ are linearly independent.
In fact, we will show that any $k = 2^\ell -1$ columns of~$\Tl$
are independent.

\begin{lemma}
\label{lem:transferMatrix}
Let $\Tl$ be the $\Ml \times \M$ matrix with $\entry{\Tl}{\a}{\b} = \binom{\b}{\a}$.
Any $2^\ell -1$ columns of matrix~$\Tl$ are linearly independent. 
\end{lemma}

\begin{proof}
Let $S \subseteq \M$ now be any set of size $k = 2^\ell -1$.
Let $\Tls$ be the $\Ml \times S$ submatrix of~$\Tl$ that consists of the columns indexed by~$S$. 
To prove the lemma we will show that
for any $0 \not= \v \in \F^{k}$ we have $\Tls \v \neq 0$. 

Let $\v = \left(v_{\a}\right)_{\a \in S}$.
Define the polynomial 
$V(\x) = \sum_{\a \in S} v_{\a} \xa~ \in \F[\x]$.
Let $V'(\x)$ be the polynomial where every variable in~$V(\x)$ is shifted by one:
$V'(\x) = V(\x + \mathbf{1})$.
From equation~(\ref{eq:shift}) we get that for any $\a \in \Ml$, 

\[
 \coeff_{V'}(\x^{\a}) = \sum_{\b \in S} \binom{\b}{\a} v_{\b} =   \row{\Tls}{\a}\, \v \,.
\]
Hence,
$\Tls \v$  gives all the coefficients of~$V'(\x)$ of support~$<\ell$.
Now it remains to show that at least one of these coefficients is nonzero. 
We show this in our next claim about {\em concentration in sparse polynomials},
which is also of independent interest.

\begin{claim}
Let $V(\x) \in \F[\x]$ be a non-zero $n$-variate 
polynomial with sparsity bounded by $2^{\ell}-1$.
Then $V'(\x) = V(\x + \mathbf{1})$ has a nonzero 
 coefficient of support~$<\ell$.
\label{cla:concSparse}
\end{claim}

We prove the claim by induction on the number of variables~$n$.
For $n = 1$,
polynomial~$V(\x)$ is univariate, i.e.\ all monomials in~$V(\x)$ have support~1.
Hence, for $\ell > 1$ it suffices to show that $V'(\x) \not= 0$.
But this is equivalent to $V(\x)\not= 0$, which holds by assumption.
If $\ell=1$,
then~$V(\x)$ is a univariate polynomial with exactly one monomial,
and therefore $V(\x+\mathbf{1})$ has a nonzero constant part. 

Now assume that the claim is true for~$n-1$ and let~$V(\x)$ have $n$ variables.
Let $\x_{n-1}$ denote the set of first $n-1$ variables. 
Let us write $V(\x) = \sum_{i=0}^d U_i \, x_n^i$, where $U_i \in \F[\x_{n-1}]$, 
for every $0 \leq i\leq d$.
Let $U'_i(\x_{n-1}) = U_i(\x_{n-1} + \mathbf{1})$ be the shifted polynomial,
for every $0 \leq i \leq d$.
We consider two cases: 

\textbf{Case 1: } 
There is exactly one index  $i \in [0, d]$ 
for which $U_i \neq 0$. Then~$U_i$ has sparsity $\le 2^\ell-1$.
Because~$U_i$ is an $(n-1)$-variate polynomial,
$U'_i$ has a nonzero coefficient of support~$<\ell$
by inductive hypothesis.

Thus, $V'(\x) = (x_n+1)^i \, U'_i$ also has a nonzero coefficient of support~$<\ell$.

\textbf{Case 2: } 
There are at least two $U_i$'s which are nonzero. 
Then there is at least one index in $i \in [0, d]$
such that~$U_i$ has sparsity $2^{\ell-1}-1$. 
And hence, by the inductive hypothesis, 
$U'_i$ has a nonzero coefficient of support~$< \ell -1$.
Consider the largest index~$j$ such that~$U'_j$ 
has a nonzero coefficient of support~$< \ell -1$.
Let the corresponding monomial be~$\x_{n-1}^\a$.
Now, as $V'(\x) = \sum_{i=0}^d U'_i \, (x_n +1)^i $, we have that 
$$
\coeff_{V'}(\x_{n-1}^{\a} x_n^j) 
= \sum_{r= j}^d \binom{r}{j} \coeff_{U'_r}(\x_{n-1}^\a).
$$
By our choice of~$j$
we have $\coeff_{U'_j}(\x_{n-1}^\a) \neq 0$ 
and $\coeff_{U'_r}(\x_{n-1}^\a) = 0$, for $r > j$.
Hence, $\coeff_{V'}(\x_{n-1}^{\a} x_n^j) \neq 0$. 
The monomial $\x_{n-1}^{\a} x_n^j$ has support~$<\ell$, which 
proves our claim and the lemma.
\end{proof}

We can use Lemma~\ref{lem:lconc} to get concentration in a 
polynomial computed by an ROABP. 
Agrawal et al.\ \cite[Lemma 8]{AGKS14} constructed a family 
$\mathcal{F} = \{\f_1(t), \f_2(t), \dots, \f_N(t)\}$ of $n$-tuples 
such that for any given polynomial $A(\x) \in \F^{w \times w}[\x]$ computed by an 
ROABP of width~$w$,
at least one of them is a basis isolating weight assignment 
and hence, provides $\log (w^2+1)$-concentration,
where 
$N = (ndw)^{O(\log n)}$.
The  degrees are bounded by
$D = \max1 \{\deg(f_{i,j}) \mid i \in [N] \text{ and } j \in [n]\} 
= (ndw)^{O(\log n)}$.
The family~$\mathcal{F}$ can be generated in time $(ndw)^{O(\log n)}$.

By Lemma~\ref{lem:lconc}, we now have an alternative PIT for {\em one\/} ROABP
because we could simply try all $\f_i \in \mathcal{F}$ for low support concentration,
and we know that at least one will work.
However,
in Lemmas~\ref{lem:blackBoxSumOfTwo} and~\ref{lem:blackBoxSumOfC}
we apply the shift to several ROABPs simultaneously,
and we have no guarantee that one of the shifts works for all of them.
We solve this problem by  combining the $n$-tuples in~$\mathcal{F}$ 
into one single shift that works for every ROABP.

Let $\L(y, t) \in \F[y, t]^n$ be the Lagrange interpolation of~$\mathcal{F}$.
That is, for all $j \in [n]$,
\[
L_j  = \sum_{i \in [N]} f_{i, j}(t) \prod_{\substack{{i' \in [N]}\\ {i' \ne i}}} \frac{y-\alpha_{i'}}{\alpha_{i}-\alpha_{i'}} \, ,\]
where $\alpha_i$ is an arbitrary unique field element associated with~$i$,
for all $i \in [N]$.
(Recall that we assume that the field~$\F$ is large enough
that these elements exist.)
Note that $L_j |_{y = \alpha_i} = f_{i,j}$.
Thus, $\L |_{y = \alpha_i} = \f_i$.
Also, $\deg_y(L_j) = N-1$ and $\deg_t(L_j) \leq D$.

\begin{lemma}
Let $A(\x)$ be a $n$-variate polynomial over a $k$-dimensional $\F$-algebra~$\A_k$
and~$\mathcal{F}$ be a family of $n$-tuples,
such that there exists
an ${\f} \in \mathcal{F}$ such that
$A'(\x, t) = A(\x + {\f}) \in \A_k(t)[\x]$  is $\ell$-concentrated.
Then, $A''(\x, y, t) = A(\x + \L) \in \A_k(y,t)[\x]$ is $\ell$-concentrated.
\label{lem:interpolation}
\end{lemma}

\begin{proof}
Let $\rank_\F\{\coeff_A(\x^\a) \mid \a \in \M\} = k'$,
for some $k' \le k$,
and $M_{\ell} = \{ \a \in \M \mid \supp(\a) < \ell \}$.
We need to show that 
$\rank_{\F(y, t)}\left\{
\coeff_{A''}(\x^\a) \mid \a \in M_{\ell}
\right\} = k'$.

Since $A'(\x)$ is $\ell$-concentrated, 
we have  that $\rank_{\F(t)}\left\{
\coeff_{A'}(\x^\a) \mid \a \in M_{\ell}
\right\} = k'$.
Recall that~$A'(\x)$ is an evaluation of~$A''$ at $y = \alpha_i$,
i.e.\ $A' (\x, t) = A''(\x, \alpha_i, t)$.
Thus, for all $\a \in \M$ we have 
$\coeff_{A'}(\xa) = \coeff_{A''}(\xa)|_{y=\alpha_i}$.

Let $C \in \F[t]^{k \times \abs{M_{\ell}}}$ be the matrix 
whose columns are~$\coeff_{A'}(\xa)$, for $\a \in M_{\ell}$.
Let similarly
$C' \in \F[y, t]^{k \times \abs{M_{\ell}}}$ be the  matrix 
whose columns are~$\coeff_{A''}(\xa)$, for $\a \in M_{\ell}$.
Then we have $C = C' |_{y = \alpha_i}$.

As $\rank_{\F(t)}(C) = k'$, there are~$k'$ rows in~$C$, say indexed by~$R$,  
such that $\det (\row{C}{R}) \ne 0$.
Because $\det (\row{C}{R})  = \det(\row{C'}{R})|_{y = \alpha_i}$,
it follows that $\det(\row{C'}{R}) \ne 0$. 
Hence, 
we have
$\rank_{\F(y, t)}(C') = k'$.
\end{proof}

Using the Lagrange interpolation, we can construct a single shift,
which works for all  ROABPs of width~$\leq w$.

\begin{theorem}
\label{thm:kronecker}
Given $n,d,w$, 
in time $(ndw)^{O(\log n)}$ one can compute a polynomial
$\f(t) \in \F[t]^n$ of degree $(ndw)^{O(\log n)}$
such that for any $n$-variate polynomial~$A(\x) \in \F^{w \times w}[\x]$
(or $\F^{1 \times w}[\x]$, or $\F[\x]$) of individual degree~$d$
that can be computed by an ROABP of width~$w$,
the polynomial~$A(\x+\f(t))$ is $\log(w^2+1)$-concentrated.
\end{theorem}

\begin{proof}
Recall that for any polynomial $A(\x) \in \F^{w \times w}[\x]$ computed by an ROABP,
at least one tuple in the family 
$\{\f_1,\f_2,\dots, \f_N\}$ obtained from~\cite[Lemma 8]{AGKS14},
gives $\log(w^2+1)$-concentration.
By Lemma~\ref{lem:interpolation}, the Lagrange interpolation~$\L(y,t)$ 
of $\{\f_1,\f_2,\dots, \f_N\}$
has~$y$- and $t$-degrees  $(ndw)^{O(\log n)}$.
After shifting an $n$-variate polynomial of individual degree~$d$ by~$\L(y,t)$,
its coefficients will be polynomials in~$y$ and~$t$, 
with degree $d'=  dn (ndw)^{O\left(\log n\right)}$.
Consider the determinant polynomial $\det(\row{C'}{R})$ from Lemma~\ref{lem:interpolation}.
As the set of coefficients of polynomial $A(\x)$ have rank bounded by $w^2$,
$\det(\row{C'}{R})$ has degree bounded by $d'' = w^2 d'$.

Note that when we replace~$y$ by~$t^{d'' +1}$,
this will not affect the non-zeroness of the determinant,
and hence, the concentration is preserved.
Thus,
$\f = \L(t^{d'' +1},t)$ is an $n$-tuple of univariate polynomials in~$t$
that fulfills the claim of the theorem.

Now, consider the case when the ROABP computes a polynomial $A(\x) \in \F^{1\times w}[x]$.
It is easy to see that there exist $S \in \F^{1\times w}$
and $B \in \F^{w\times w}[\x]$ computed by a width-$w$ ROABP such that
$A = SB$. 
We know that $B(\x +\f(t))$ has $\log(w^2+1)$-concentration. 
As multiplying by $S$ is a linear operation, one can argue 
as in the proof of Lemma~\ref{lem:vectorConc} that any linear dependence among 
coefficients of $B(\x + \f(t))$ also holds among coefficients of $A(\x + \f(t))$.
Hence, $A(\x + \f(t))$ has $\log(w^2+1)$-concentration. 
A similar argument would work when $A(\x) \in \F[\x]$,
by writing $A = SBT$, for some $S \in \F^{1\times w}$ and $T \in \F^{w\times 1}$.
\end{proof}


\section{Discussion}
The first question is whether one can make the time complexity
for PIT for the sum of~$c$ ROABPs
proportional to~$w^{O(c)}$ instead of~$w^{O(2^c)}$. 
This blow up happens because, when we want to combine $w+1$ partial
derivative polynomials given by ROABPs of  width~$w$, 
we get an ROABP of width~$O(w^2)$. 
There are examples where this bound seems tight. So,
a new property of sum of ROABPs needs to be discovered.

It also needs to be investigated if these ideas can be generalized to
work for sum of more than constantly many ROABPs, or depth-$3$ multilinear circuits.

As mentioned in the introduction, the idea for equivalence of two ROABPs
was inspired from the equivalence of two read once boolean branching programs (OBDD). 
It would be interesting to know if there are concrete connections between
arithmetic and boolean branching programs. 
In particular, can ideas from identity testing of an ROABP be applied
to construct pseudo-randomness for OBDD.
E.g.\ the less investigated model, XOR of constantly many OBDDs
can be checked for unsatisfiability by modifying our techniques.


\section{Acknowledgements}
We thank Manindra Agrawal, Chandan Saha and Vineet Nair for very useful discussions and constant encouragement.
The work was initiated when TT was visiting CSE, IIT Kanpur. 
Part of the work was done during Dagstuhl Seminar 14391 on Algebra in Computational Complexity 2014.
We thank anonymous referees for the useful suggestions. 
\bibliographystyle{alpha}
\bibliography{sumROABP}

\end{document}